\documentclass[english,onecolumn,draftcls]{IEEEtran}
 
\makeatletter

\usepackage[T1]{fontenc}
\usepackage[latin9]{inputenc}
\usepackage{amsthm}
\usepackage{amsmath}
\usepackage{amssymb}
\usepackage{graphicx}
\usepackage{dsfont}
\usepackage{cite}
\usepackage{amsmath}
\usepackage{array}
\usepackage{multirow}
\usepackage{caption}
\usepackage{color}

\usepackage{multicol}

\theoremstyle{plain}

\theoremstyle{plain}

\theoremstyle{plain}
\newtheorem{lem}{\protect\lemmaname}
\theoremstyle{plain}
\newtheorem{thm}{\protect\theoremname}
\theoremstyle{plain}
  
\theoremstyle{definition}
\newtheorem{defn}{\protect\definitionname}
\theoremstyle{definition}

\theoremstyle{definition}

\makeatother
  
\usepackage{babel} 

\providecommand{\claimname}{Claim}
\providecommand{\lemmaname}{Lemma}
\providecommand{\propositionname}{Proposition}
\providecommand{\theoremname}{Theorem}
\providecommand{\corollaryname}{Corollary} 
\providecommand{\definitionname}{Definition}
\providecommand{\assumptionname}{Assumption}
\providecommand{\remarkname}{Remark}

\newcommand{\poly}{{\rm poly}}

\newcommand{\svtil}{\tilde{\mathbf{s}}}
\newcommand{\stil}{\tilde{s}}
\newcommand{\Sctil}{\widetilde{\mathcal{S}}}
\newcommand{\yvtil}{\tilde{\mathbf{y}}}
\newcommand{\Kchat}{\widehat{\mathcal{K}}}

\newcommand{\Bernoulli}{\mathrm{Bernoulli}}

\newcommand{\pe}{P_{\mathrm{e}}}

\newcommand{\uv}{\mathbf{u}}

\newcommand{\xv}{\mathbf{x}}

\newcommand{\Xv}{\mathbf{X}}
\newcommand{\yv}{\mathbf{y}}
\newcommand{\Yv}{\mathbf{Y}}

\newcommand{\Ac}{\mathcal{A}}
\newcommand{\Bc}{\mathcal{B}}
\newcommand{\Cc}{\mathcal{C}}

\newcommand{\Ic}{\mathcal{I}}
\newcommand{\Kc}{\mathcal{K}}

\newcommand{\Sc}{\mathcal{S}}

\newcommand{\EE}{\mathbb{E}}
\newcommand{\PP}{\mathbb{P}}
\newcommand{\RR}{\mathbb{R}}

\newcommand{\Lc}{\mathcal{L}}
\newcommand{\cv}{\mathbf{c}}

\newcommand{\sv}{\mathbf{s}}

\usepackage{hyperref} 
\usepackage{enumitem}
\usepackage{float}
\usepackage{algorithm}
\usepackage{algpseudocode}

\providecommand{\tabularnewline}{\\}
\floatstyle{ruled}
\newfloat{algorithm}{tbp}{loa}
\providecommand{\algorithmname}{Algorithm}
\floatname{algorithm}{\protect\algorithmname}

\setenumerate[1]{label=\arabic*.}

\usepackage{geometry} 
\makeatletter
\newcommand{\manuallabel}[2]{\def\@currentlabel{#2}\label{#1}}
\makeatother

\begin{document} 

\title{Sublinear-Time Non-Adaptive Group Testing with $O(k \log n)$ Tests via Bit-Mixing Coding}

\author{Steffen Bondorf, Binbin Chen, Jonathan Scarlett, Haifeng Yu, Yuda Zhao}
\maketitle

\begin{abstract}
    The group testing problem consists of determining a small set of defective items from a larger set of items based on tests on groups of items, and is relevant in applications such as medical testing, communication protocols, pattern matching, and many more.  While rigorous group testing algorithms have long been known with runtime at least linear in the number of items, a recent line of works has sought to reduce the runtime to $\poly(k \log n)$, where $n$ is the number of items and $k$ is the number of defectives. 
    In this paper, we present such an algorithm for non-adaptive {probabilistic} group testing termed {\em bit mixing coding} (BMC), {which builds on techniques that encode item indices in the test matrix, while incorporating novel ideas based on erasure-correction coding}. We show that BMC achieves asymptotically vanishing error probability with $O(k \log n)$ tests and $O(k^2 \cdot \log k \cdot \log n)$ runtime, in the limit as $n \to \infty$ (with $k$ having an arbitrary dependence on $n$).  This closes a recently-proposed open problem of simultaneously achieving $\poly(k \log n)$ decoding time using $O(k \log n)$ tests without any assumptions on $k$.  In addition, we show that the same scaling laws can be attained in a commonly-considered noisy setting, in which each test outcome is flipped with constant probability.
\end{abstract}

\long\def\symbolfootnote[#1]#2{\begingroup\def\thefootnote{\fnsymbol{footnote}}\footnote[#1]{#2}\endgroup}

\symbolfootnote[0]{ 

The authors of this paper are alphabetically ordered.  This work was done while S.~Bondorf and Y.~Zhao were at the National University of Singapore.

S.~Bondorf is with NTNU Trondheim, Norway. (e-mail: steffen.bondorf@ntnu.no).

B. Chen is with the Advanced Digital Sciences Center, Singapore (e-mail: binbin.chen@adsc-create.edu.sg).

J.~Scarlett is with the Department of Computer Science \& Department of Mathematics, National University of Singapore (e-mail: scarlett@comp.nus.edu.sg). 

H.~Yu is with the Department of Computer Science, National University of Singapore (e-mail: haifeng@comp.nus.edu.sg).

Y.~Zhao is with Advance.AI (e-mail: yudazhao@gmail.com).

A preliminary conference version of this work was presented at the International Conference on Information Processing in Sensor Networks (IPSN), Montreal, 2019 \cite{Bon19}.

}

%
%
\section{Introduction} \label{sec:intro}

The group testing problem consists of determining a small subset of defective items within a larger set of items, based on a tests performed on groups of items, and corresponding outcomes that indicate whether the group contains at least one defective item. This problem has a history in medical testing \cite{Dor43}, and has regained significant attention following new applications in areas such as communication protocols \cite{Ant11}, pattern matching \cite{Cli10}, and database systems \cite{Cor05}, and connections with compressive sensing \cite{Gil08,Gil07}.  The design and analysis of group testing algorithms remains an active ongoing area of research; see \cite{Du93,Ald19} for comprehensive surveys.  {Some of the key defining features of the group testing problem are outlined as follows:
\begin{itemize}
    \item {\bf Combinatorial vs.~probabilistic.}  In {\em combinatorial group testing} \cite{Du93}, one seeks to construct a testing procedure that guarantees the recovery of {\em all} defective sets up to a certain size.  In contrast, in {\em probabilistic group testing} \cite{Ald19}, the test design may be randomized, and the algorithm is allowed some non-zero probability of error.  Combinatorial group testing is also known as the {\em for-all model} or the {\em zero-error recovery criterion}, and probabilistic group testing is also known as the {\em for-each model} or the {\em small-error recovery criterion}. 
    \item {\bf Adaptive vs.~non-adaptive.} In the {\em adaptive} setting, each test may be designed based on all previous outcomes, whereas in the {\em non-adaptive setting}, all tests must be chosen prior to observing any outcomes.  The non-adaptive setting is often preferable in practice, as it permits the tests to be performed in parallel. 
    \item {\bf Noiseless vs.~noisy.} In the {\em noiseless} setting, the test outcomes are perfectly reliable, whereas in {\em noisy settings}, some tests may be flipped according to some probabilistic or adversarial noise model.
\end{itemize}
Our focus is on non-adaptive probabilistic group testing; we formally introduce the noiseless model below, and turn to the noisy setting in Section \ref{sec:noisy}.
}


\subsection{Problem Setup} \label{sec:setup}

The group testing problem consists of $n$ items labeled $\{1,\dotsc,n\}$, a subset $\Kc$ of which are {\em defective}.  We seek to identify $\Kc$ via a series of suitably-chosen tests.  Except where stated otherwise, we consider the noiseless setting, in which each test takes the form
\begin{equation}
    Y = \bigvee_{j \in \Kc} X_j, \label{eq:gt_noiseless_model}
\end{equation}
where the test vector $X = (X_1,\dotsc,X_n) \in \{0,1\}^n$ indicates which items are included in the test, and $Y \in \{0,1\}$ is the resulting test outcome.  That is, the output indicates whether at least one defective item is included in the test.   The goal is to design a sequence of tests $X^{(1)},\dotsc,X^{(t)}$, with $t$ ideally as small as possible, such that the outcomes can be used to reliably recover the defective set $\Kc$.  

We focus on {\em non-adaptive} test designs, in which all tests must be chosen prior to observing any outcomes.  Accordingly, the tests $X^{(1)},\dotsc,X^{(t)}$ are represented by a {\em test matrix} $\Xv \in \{0,1\}^{t \times n}$ whose $i$-th column is $X^{(i)} \in \{0,1\}^t$.  The corresponding test outcomes are denoted by $\Yv = (Y^{(1)},\dotsc,Y^{(t)})$, with $Y^{(i)} \in \{0,1\}$ generated from $X^{(i)}$ according to the model \eqref{eq:gt_noiseless_model}.
  
Given the tests and their outcomes, a \emph{decoder} forms an estimate $\Kchat$ of $\Kc$.  We consider the exact recovery criterion, in which the error probability is given by 
\begin{equation}
    \pe := \PP[\Kchat \ne \Kc]. \label{eq:pe}
\end{equation}
We assume that $|\Kc| \le k$ for some $k$ that is known to the group testing algorithm.  That is, the algorithm knows an upper bound on $|\Kc|$ is known but not necessarily the exact value.  {This is a standard assumption in the literature, and an assumption of this kind is necessary in the non-adaptive setting -- without an upper bound on $|\Kc|$, one would need to account for scenarios such as $|\Kc| \ge \frac{n}{2}$ that require $n$ tests \cite{Ald18}.}

Our analysis will hold for an arbitrary {\em fixed} defective set $\Kc$ with cardinality at most $k$, meaning that the probability in \eqref{eq:pe} is only with respect to our randomized test design $\Xv$.  However, we can alternatively view our results as certifying the existence of a fixed matrix $\Xv$ yielding small $\pe$ with respect a randomly generated $\Kc$ whose distribution is {\em independent of $\Xv$} and satisfies $|\Kc| \le k$ almost surely.

{Throughout the paper, we use the standard asymptotic notation $O(\cdot)$, $o(\cdot)$, $\Theta(\cdot)$, $\Omega(\cdot)$ and $\omega(\cdot)$.}

\subsection{Summary of Results} \label{sec:contr}

{The vast majority of the group testing literature has sought to develop test designs with as few tests as possible, and decoding algorithms whose runtime is linear or polynomial in the number of items.  Recently, however, a line of works has developed test designs and decoding algorithms that permit more efficient $\poly(k \log n)$ decoding time when there are $n$ items and $k$ defectives, thereby considerably reducing the dependence on $n$.  This was first done in the combinatorial  setting \cite{Che09,Ind10,Ngo11} and more recently in the probabilistic setting \cite{Cai13,Lee15a,Ina19}; see Section \ref{sec:related} for details.  
}  

{In this paper, we introduce a non-adaptive probabilistic group testing procedure termed {\em bit mixing coding} (BMC) that attains asymptotically vanishing error probability as $n \to \infty$ with $O(k \log n)$ tests and $O(k^2 \cdot \log k \cdot \log n)$ decoding time.  The $O(k \log n)$ number of tests is known to be order-optimal whenever $k \le O(n^{1-\epsilon})$ for some $\epsilon > 0$. BMC is the first algorithm to achieve such optimal number of tests together with $\poly(k \log n)$ decoding time, resolving an open problem recently posed in \cite{Ina19}. As we will see in Section \ref{sec:related}, the best known previous approach with $\poly(k \log n)$ decoding time needed to use $O(k \cdot \log k \cdot \log n)$ tests \cite{Cai13,Lee15a}.
In the terminology of \cite{Bal13,Ald14a,Ald19}, order-optimality in the number of tests amounts to attaining a {\em positive rate}: The number of bits learned per test is $\Theta(1)$, whereas existing algorithms that use $O(k \cdot \log k \cdot \log n)$ tests \cite{Cai13,Lee15a} only learn $O\big( \frac{1}{\log k} \big)$ bits per test.
%
%
Finally, we note that the $O(k^2 \cdot \log k \cdot \log n)$ decoding time of BMC falls short of the $O(k \cdot \log k \cdot \log n)$ decoding time achieved using $O(k \cdot \log k \cdot \log n)$ tests \cite{Cai13,Lee15a}, which leaves open the possibility of reducing our runtime further while maintaining order-optimality in the number of tests.

%




BMC has a few additional salient features. While BMC uses a randomized test design, along the way we provide sufficient conditions for success that hold with high probability and can be verified in time $\poly(k \log n)$.
BMC also can naturally incorporate mechanisms for combating noise in the test outcomes: In Section \ref{sec:noisy}, we describe straightforward modifications to the test design and decoding algorithm to permit randomly flipped test outcomes while preserving the guarantees on the number of tests and decoding time.}
 
\section{Related Work} \label{sec:related}

In this section, we provide a detailed overview of the most related existing works, first focusing on the theoretical results and then discussing the corresponding algorithmic ideas.

\subsection{Overview of Existing Group Testing Results} \label{sec:ex_gt}

{The existing literature on non-adaptive group testing most related to this work is summarized in Table \ref{tbl:summary}.  Along with the distinction between the combinatorial and probabilistic settings, we highlight the following features of the test designs and recovery algorithms:
\begin{itemize}
    \item {\bf Explicit vs.~randomized.} Many of the tightest bounds in the literature are based on {\em randomized} test designs.  In contrast, an efficient deterministic procedure for constructing a test design is said to be {\em explicit}.  There are various notions of how efficient the procedure should be to warrant this terminology \cite{Por11}; to facilitate our discussion, we only consider the most lenient notion in the literature, requiring the test matrix $\Xv \in \{0,1\}^{t \times n}$ can be deterministically constructed in time polynomial in $t$ and $n$.
    \item {\bf Decoding efficiency.} The majority of the group testing literature considers algorithms with $\Omega(n)$ runtime, e.g., as a result of traversing the entire matrix $\Xv \in \{0,1\}^{t \times n}$.  Our focus, however, is on decoding algorithms with a significantly lower runtime of the form $\poly(k \log n)$ decoding time, ideally with a low polynomial power.  Such algorithms attain {\em sublinear-time decoding} (i.e., decoding time scaling as $o(n)$) when $k$ grows sufficiently slowly with respect to $n$.
    \item {\bf Recovery criteria.} Except where stated otherwise, all results that we overview correspond to the exact recovery criterion, requiring that $\Kchat = \Kc$ (see \eqref{eq:pe}).  However, we also briefly mention two other recovery criteria appearing in Table \ref{tbl:summary}: (i) The {\em list decoding} criterion \cite{Deb05,Dya83,Ras90,Che09,Ind10,Ngo11} only requires identifying a superset of the defective set, typically constrained to be of size $O(k)$; (ii) The {\em approximate recovery} criterion \cite{Lee15a,Sca17} only requires identifying a fraction $1-\epsilon$ of the defectives, for some constant $\epsilon > 0$.
\end{itemize}
We proceed by discussing the results in Table \ref{tbl:summary} in more detail; the most relevant algorithmic ideas used in attaining these results will be discussed in Section \ref{sec:techniques}.

{\bf Combinatorial group testing.} The combinatorial setting poses a strictly harder problem than the probabilistic setting, in the sense of requiring $t = \Omega\big( \min\big\{ k^2 \frac{\log n}{\log k}, n \big\} \big)$ tests \cite{Dya82} as opposed to $O(k \log n)$ \cite{Ald14a}.  The best-known $O(k^2 \log n)$ upper bound on the number of tests was originally attained with $\Omega(n)$ decoding time, first using random coding methods \cite{Dya82} and then using explicit designs \cite{Por11}.  See \cite{Kau64,Dya83,Dya00,Kim04} and the references therein for further related works.

More recently, algorithms were developed that attain ${\rm poly}(t)$ decoding time \cite{Che09,Ind10,Ngo11}, with Cheragchi \cite{Che09} focusing on list decoding, and Indyk {\em et al.}~\cite{Ind10} and Ngo {\em et al.}~\cite{Ngo11} considering exact recovery.  In particular, the latter works showed that the decoding time can be reduced to ${\rm poly}(t)$ while maintaining the $t = O(k^2 \log n)$ scaling of \cite{Por11}.  To achieve this, \cite{Ind10} used a randomized design, and \cite{Ngo11} presented an explicit construction.

A more recent work provided an explicit construction attaining $t = O(k^2 \log^2 n)$ with $O(k^3 \log^2 n)$ decoding time \cite{Che19}.  The main feature highlighted in \cite{Che19} is the simplicity of the construction, but a drawback is an additional $\log n$ factor in the number of tests.

\begin{table}
\begin{center}
    \begin{tabular}{|>{\centering}m{2.8cm}|>{\centering}m{2.35cm}|>{\centering}m{2.9cm}|>{\centering}m{3.0cm}|>{\centering}m{3.5cm}|}
    \hline 
    \textbf{References} & \textbf{Guarantee} & \textbf{Number of tests $t$} & \textbf{Runtime} & \textbf{Construction}\tabularnewline
    \hline 
    \hline
    {\em Lower Bound} \cite{Dya82} & Combinatorial & $\Omega\big(\min\big\{k^2 \frac{\log n}{\log k},n\big\}\big)$ & - & -\tabularnewline
    \hline
    D'yachkov-Rykov \cite{Dya83} & Combinatorial & $O(k^{2}\log n)$ & $\Omega(n)$ & Randomized\tabularnewline
    \hline 
    Kautz-Singleton \cite{Kau64} & Combinatorial & $O\big(k^{2}\frac{\log^2 n}{\log^2 k}\big)$ & $\Omega(n)$ & Explicit\tabularnewline
    \hline 
    Porat-Rothschild \cite{Por11} & Combinatorial & $O(k^{2}\log n)$ & $\Omega(n)$ & Explicit\tabularnewline
    \hline 
    Cheragchi \cite{Che09}  & Combinatorial
    
    {\bf (list decoding only)} & $O(k\cdot2^{\log^{3}\log n})$ & ${\rm poly(t)}$ & Explicit\tabularnewline
    \hline 
    Indyk \emph{et al.}~\cite{Ind10} & Combinatorial & $O(k^{2}\log n)$ & ${\rm poly}(k)t\log^{2}(t)+O(t^{2})$ & Randomized

    {\footnotesize (Explicit if $k = O\big(\frac{\log n}{\log \log n}\big)$)}\tabularnewline
    \hline 
    Ngo \emph{et al.}~\cite{Ngo11} & Combinatorial & $O(k^{2}\log n)$ & ${\rm poly(t)}$ & Explicit\tabularnewline
    \hline
    Cheraghchi-Ribeiro~\cite{Che19} & Combinatorial & $O(k^{2}\log^2 n)$ & $O(k^3 \log^2 n)$ & Explicit\tabularnewline
    \hline
    \hline
    {\em Lower Bound} \cite{Mal78} & Probabilistic & $\Omega\big( k \log \frac{n}{k} \big)$ & - & -\tabularnewline
    \hline 
    Various \cite{Cha11,Che11,Mal13,Ald14a,Sca15b,Sca17b,Coj19,Coj19a} & Probabilistic & $O(k \log n)$ & $\Omega(n)$ & Randomized\tabularnewline
    \hline
    Mazumdar~\cite{Maz16} & Probabilistic & $O\big(k \frac{\log^2 n}{\log k}\big)$ & $\Omega(n)$ & Explicit\tabularnewline
    \hline
    Inan \emph{et al.}~\cite{Ina19} & Probabilistic & $O(k \log n)$ & $\Omega(n)$ & Explicit\tabularnewline
    \hline
    GROTESQUE~\cite{Cai13} & Probabilistic & $O(k\cdot\log k\cdot\log n)$ & $O(k\cdot\log k\cdot\log n)$ & Randomized\tabularnewline
    \hline 
    SAFFRON \cite{Lee15a}  & Probabilistic & $O(k\cdot\log k\cdot\log n)$ & $O(k\cdot\log k\cdot\log n)$ & Randomized\tabularnewline
    \hline
    SAFFRON \cite{Lee15a}  & Probabilistic 

    {\bf (approximate recovery only)} & $O(k \log n)$ & $O(k \log n)$ & Randomized\tabularnewline 
    \hline 
    Inan \emph{et al.}~\cite{Ina19} & Probabilistic & $O\big(k\cdot\log n\cdot\log\frac{\log n}{\log k}\big)$ & $\ensuremath{O\big(k^{3}\cdot\log n\cdot\log\frac{\log n}{\log k}\big)}$ & Explicit\tabularnewline
    \hline 
    \textbf{This paper} & Probabilistic & \textbf{$O(k\log n)$} & \textbf{$O(k^2 \cdot\log k\cdot\log n)$} & Randomized\footnotemark\tabularnewline
    \hline
    \end{tabular}
    \par\end{center}

    \caption{{Overview of existing non-adaptive group testing results, with $n$ items, $k$ defectives, and $t$ tests.  Two of the works listed above attain a reduced number of tests and/or runtime by considering a less stringent recovery criterion than exact recovery: (i) In \cite{Che09}, a list $\Lc$ of size $O(k)$ is returned, and it is only required that $\Kc \subseteq \Lc$; (ii) In \cite{Lee15a}, a set of size $k$ is returned, but it is only required to contain a fraction $1-\epsilon$ of the defectives for some $\epsilon > 0$, and the scaling laws shown do not apply in the limit as $\epsilon \to 0$.}\label{tbl:summary}}
\end{table}



{\bf Probabilistic group testing.} In the probabilistic setting, the $\Omega\big( k \log\frac{n}{k}\big)$ lower bound \cite{Mal78} on the number $t$ of tests indicates that the scaling of $t = O(k \log n)$ is optimal whenever $k \le O(n^{1-\epsilon})$ for some constant $\epsilon > 0$.   Under a randomized test design and with $\Omega(n)$ decoding time, numerous results attaining asymptotically vanishing error probability with $t = O(k \log n)$ have been obtained \cite{Cha11,Che11,Mal13,Ald14a,Sca15b,Sca17b,Coj19,Coj19a}.  In addition, Inan {\em et al.}~\cite{Ina19} attain $t = O(k \log n)$ with $\Omega(n)$ decoding time using an explicit design, improving on an earlier $t = O\big( k \frac{\log^2 n}{\log k} \big)$ bound due to Mazumdar \cite{Maz16}.

\footnotetext{{Despite being randomized, we also provide sufficient conditions that hold with high probability, that ensure success, and that can be verified in time $\poly(k \log n)$; see the discussion following Lemma \ref{lem:lcs}.}}

The most relevant existing works to this paper are those attaining $\poly(k \log n)$ decoding time in the probabilistic setting, particularly GROTESQUE \cite{Cai13}, SAFFRON \cite{Lee15}, and the study of the Kautz-Singleton construction by Inan {\em et al.}~\cite{Ina19}.  GROTESQUE and SAFFRON attain $O(k \cdot \log k \cdot \log n)$ for both the number of tests and runtime using randomized designs, whereas \cite{Ina19} attains $t = O\big(k \cdot \log n \cdot \log \frac{\log n}{\log k}\big)$ with $O\big(k^3 \cdot \log n \cdot \log \frac{\log n}{\log k}\big)$ decoding time using an explicit design.  SAFFRON additionally improves the $O(k \cdot \log k \cdot \log n)$ scaling to $O( C(\epsilon) k \log n)$ (for some $C(\epsilon) > 0$) under an approximate recovery criterion that only requires $(1-\epsilon) k$ defectives to be identified. But an inspection of the proof reveals that $C(\epsilon) \to \infty$ as $\epsilon \to 0$, precluding exact recovery with an order-optimal number of tests.

We note that if $k = \Theta(n^{\alpha})$ for some constant $\alpha \in (0,1)$ and one merely requires $t = O(k \log n)$ and $\poly(k \log n)$ decoding time, then there exist several algorithms achieving this goal as a result of having $O(n t)$ decoding time \cite{Cha11,Che11,Ald14a,Sca17b}, or $O(k^3 \cdot \log n)$ runtime in the case of \cite{Ina19}.  We therefore contend that the regime of primary interest for seeking $\poly(k \log n)$ decoding time is the sparser regime in which $k = o(n^{\alpha})$ for any constant $\alpha > 0$. 


{\bf Comparison with our results.} As outlined above, our main contribution is to bring the number of tests down to the optimal $t = O(k \log n)$ scaling while maintaining efficient decoding time (namely, $O(k^2 \log k \cdot \log n)$), and imposing no restrictions on $k$.

}

\subsection{Overview of Existing Group Testing Techniques} \label{sec:techniques}

{
The above-outlined group testing algorithms with efficient decoding are predominantly either based on {\em code concatenation}, or utilize the idea of {\em encoding item indices} into the test matrix.  While these are two fundamentally different techniques, they in fact share a common high-level structure, depicted in Figure \ref{fig:concat_mask}.  Initially, a matrix is formed with $n$ rows, denoted by $\cv_1,\dotsc,\cv_n$.  The final (transpose of the) test matrix is formed by expanding each entry of each $\cv_i$ to a longer binary sequence:
\begin{itemize}
    \item In the concatenated coding approach, the ``codewords'' $\cv_1,\dotsc,\cv_n$ form a non-binary {\em outer code}, and each codeword symbol is mapped to a block in the final test matrix via an {\em inner code}.  For instance, the Kautz-Singleton construction \cite{Kau64} employs a trivial inner code (along with a Reed-Solomon outer code) that maps to a vector with one in a single entry indexing the corresponding non-binary symbol, and zeros elsewhere.
    \item In the alternative approach that encodes item indices, the entries of $\cv_1,\dotsc,\cv_n$ are binary.  Any zero entry is trivially mapped to a block of zeros, whereas the entries equaling one in $\cv_i$ are mapped to a binary vector describing the item's index, $i \in \{1,\dotsc,n\}$.  A standard binary description would require exactly $\lceil \log_2 n \rceil$ bits, but a longer length may be used to facilitate fast decoding \cite{Lee15a} and/or improve robustness to noise \cite{Cai13}.
\end{itemize}
We proceed by discussing each of these in more detail.  (In Section \ref{sec:cs}, we also discuss the origins of the latter approach in the context of compressive sensing.)

\begin{figure}
  \centering
  \includegraphics[width=0.95\columnwidth]{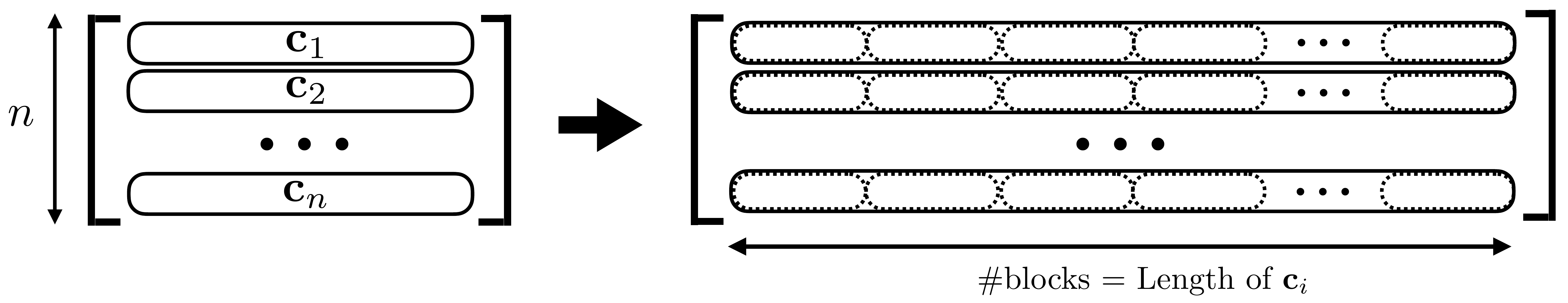}
  \caption{{High-level illustration of the (transpose of the) group testing matrices formed by code concatenation or techniques that directly encode the items' indices; the details of the two techniques differ, and are given in the text.}}
  \label{fig:concat_mask}
\end{figure}

{\bf Code concatenation.} Early uses of concatenated codes in combinatorial group testing incurred $\Omega(n)$ decoding time \cite{Kau64,Por11}.  To achieve $\poly(k \log n)$ time decoding, the main idea employed in \cite{Ind10,Ngo11} is {\em list decoding}: An {inner code} of length $O(k \log n)$ is used that in itself would suffice the attain the above-mentioned list decoding criterion, and an {outer code} of rate $O\big( \frac{1}{k} \big)$ is used to efficiently resolve the remaining uncertainty in the list (e.g., a Reed-Solomon code suffices), leading to $O(k^2 \log n)$ tests in total.
In addition, \cite{Ngo11} proposed a novel recursive construction that allows the decoder to recover a ${\rm poly}(k)$-size superset of the defective set, after which a standard decoding strategy for disjunct matrices can be used to resolve the remaining uncertainty.  

The recent work of Inan {\em et al.}~\cite{Ina19} follows the Kautz-Singleton construction \cite{Kau64} (i.e., a Reed-Solomon outer code and trivial inner code), but varies its parameters (e.g., the code length) to better suit the probabilistic group testing setting, as opposed to the combinatorial setting considered in \cite{Kau64}.

{\bf Techniques that encode item indices.}  Existing group testing techniques that encode the items' indices into the test matrix have focused predominantly on the probabilistic setting, rather than the combinatorial setting.\footnote{An exception is \cite{Che19}, which will briefly be discussed in Section \ref{sec:cs}.}  In particular, we highlight the GROTESQUE \cite{Cai13} and SAFFRON \cite{Lee15a} algorithms, both of which use random binary strings $\cv_1,\dotsc,\cv_n$ in Figure \ref{fig:concat_mask} drawn from a suitably-designed distribution.  This distribution is chosen such that, with high probability, each defective item $i \in \Kc$ is {\em isolated}:  There exists an index $j$ for which the $j$-th entry of $\cv_i$ is $1$, but the $j$-th entry of each string in $\{\cv_{i'}\}_{i' \in \Kc \setminus \{i\}}$ is $0$.\footnote{This discussion is based on the ``singleton-only'' version of SAFFRON.  The general version also makes use of blocks with two defectives, which are useful after one of them has already been identified.  However, both versions yield the same scaling laws in the number of tests and decoding time.}  When this property holds, each test outcome in the corresponding block (see the right of Figure \ref{fig:concat_mask}) is positive if and only if item $i$ is included in the test.

The two algorithms primarily differ in (i) how to locate the blocks corresponding to isolated defectives, and (ii) how to decode the defective items' indices:
\begin{itemize}
    \item In SAFFRON, the two are done simultaneously by using blocks of length $2 \lceil\log_2 n\rceil$ to encode each item's index and its complement.  Then, it is shown that each block corresponds to an isolated defective if and only if the corresponding outcomes contain exactly $\lceil \log_2 n \rceil$ ones.  In addition, when this is the case, the item index can be directly read from the first $\lceil \log_2 n \rceil$ outcomes.
    \item GROTESQUE uses a procedure termed {\em multiplicity testing}, in which items are randomly tested, and the cases ``no defectives'' vs.~``one defective'' vs.~``two or more defectives'' can be distinguished by simply counting the number of $1$'s in the outcomes.  To ensure the reliable recovery of the item's index in the case that the answer is ``one defective'', the index is encoded in each block using an expander code to combat possible noise, though a trivial length-$\lceil \log _2 n\rceil$ code would also suffice in the noiseless setting.
\end{itemize}
In both algorithms, the vectors $\cv_1,\dotsc,\cv_n$ in Figure \ref{fig:concat_mask} are chosen to have length $O(k \log k)$.  Since it is unknown in advance which blocks of tests will correspond to isolated defectives, every block in the final test matrix must incur $O(\log n)$ tests, for a total of $t = O(k \cdot \log k \cdot \log n)$.

{\bf Comparison with our techniques.} Our group testing strategy, BMC, is outlined in Section \ref{sec:overview}, but at this point we can already highlight some of the key differences to the techniques described above:
\begin{itemize}
    \item While the strings $\cv_1,\dotsc,\cv_n$ are mutually independent in \cite{Cai13,Lee15a}, this is far from being true in BMC.  Instead, we independently generate a {\em much smaller} number of strings, and then let each $\cv_i$ equal one of these strings selected uniformly at random.  Hence, there are a large number of {\em repeated strings}; we only seek to ensure that there are no repetitions {\em among the defectives}.
    \item The first step of our decoding algorithm is to {\em identify the strings associated with defectives}, whereas in \cite{Cai13,Lee15a} the goal is to {\em identify the blocks corresponding to isolated defectives}.  These are distinct goals, and are solved using different techniques: In contrast with the above-outlined approaches used by GROTESQUE and SAFFRON, we can achieve our goal by performing a simple one-by-one check on the small set of strings mentioned in the previous dot point.
    \item We not only seek for each defective item $i \in \Kc$ to have a single isolated index in $\cv_i$, but rather, $O(\log n)$ of them.  This may sound like a more restrictive condition that potentially {\em increases} the number of tests, but it is made up for by the following crucial observation: We do {\em not} blow up the number of tests by a factor of $O(\log n)$ in order to ensure $O(\log n)$ ``collision-free'' tests for each defective item.  Instead, we treat any collisions as erasures, and control for them using erasure-correcting coding.\footnote{GROTESQUE \cite{Cai13} employs expander codes to combat {\em random noise}, but this is a distinct notion to our idea of using erasure-correction to combat collisions, and the former technique does not transfer readily to the latter.}  Hence, instead of seeking $O(\log n)$ specific collision-free tests, we allow the defectives to {\em share the damage of collisions} in a controlled manner.
\end{itemize}
We also briefly contrast BMC with the list-decoding approach \cite{Deb05,Ngo11}, which first finds a ``small enough'' superset of $\Kc$ (e.g., of size $O(k)$ or $O({\rm poly}(k))$), and uses further tests to resolve the false positives.  The first decoding step of BMC finds up to $k$ masking strings, and the union of sets of items associated with those masking strings is a superset of $\Kc$ with high probability.  However, the number of items assigned to each masking string is in fact very large, leading to this superset having size $O\big( \frac{n}{k \log k} \big)$.  In sparse settings (e.g., $k = O({\rm poly}(\log n))$), this size far exceeds ${\rm poly}(k)$, indicating that BMC is fundamentally different to list decoding.}

\subsection{Related Techniques for Compressive Sensing} \label{sec:cs} 

{While we focused on the most related group testing works when discussing the idea of encoding items' indices into $\Xv$, closely-related techniques appeared prior to those works in the context of {\em compressive sensing} (CS) \cite{Cor06,Gil06,Gil07,Ber08a,Che17,Che19}.  In this problem, the goal is to design a (real-valued) measurement matrix $\Xv \in \RR^{t \times n}$ that permits the recovery of $k$-sparse vectors $\beta \in \RR^n$ via linear measurements of the form $\yv = \Xv\beta$ (or similarly with noise added).  

Group testing can be viewed as a Boolean counterpart to CS  \cite{Gil08,Ati12}, but there are also important differences between the two.  In particular, in contrast with CS, group testing is inherently non-linear due to the ``OR'' operation.  As a result, several decoding techniques used in compressive sensing cannot be used in group testing , notably including the idea of ``subtracting off'' previously-found values in the sparse vector.  Due to these differences, CS results often differ significantly from group testing.  For instance, it is possible to attain the ``for-all'' guarantee in CS with $t = O(k \log n)$ measurements \cite{Can08}, in stark contrast with the $\Omega\big(\min\big\{k^2 \frac{\log n}{\log k},n\big\}\big)$ lower bound for combinatorial group testing \cite{Dya82}.

An early CS work of Cormode and Muthukrishnan \cite{Cor06} followed the structure of Figure \ref{fig:concat_mask}, utilizing a disjunct matrix in the first step.  This approach could readily be applied to group testing, but would not be suited to the probabilistic setting, due to the $\Omega\big(\min\big\{k^2 \frac{\log n}{\log k},n\big\}\big)$ number of rows required for a disjunct matrix.  The number of tests was subsequently reduced using test designs based on random selection \cite{Gil06,Gil07}, an idea also used the above-outlined group testing works \cite{Cai13,Lee15a}.  In fact, the CS designs in \cite{Gil06,Gil07} use more sophisticated random selection techniques based on random binning with variable bin sizes, but these appear to be less suited to group testing due to them strongly exploiting the linearity of the measurements as discussed above.  

More recent CS works utilized alternative constructions based on expanders and extractors \cite{Ber08a,Che17,Che19}.  In addition, \cite{Che19} showed that constructions of this kind can also be applied to combinatorial group testing, though we are not aware of any similar attempts for probabilistic group testing.

Despite these advances in the context of compressive sensing, we are not aware of any construction that can be adapted to provide a probabilistic group testing algorithm with $O(k \log n)$ tests and $\poly(k \log n)$ decoding time.  Essentially, each of these works appears to exhibit one or both of the following roadblocks: (i) the number of tests is inherently limited to behave as $\Omega(k \cdot \log k \cdot \log n)$ or higher, thus failing to improve on \cite{Cai13,Lee15a}; (ii) the decoding procedure crucially exploits the linearity in the compressive sensing model. 

Finally, to our knowledge, none of the existing CS works include the unique aspects of BMC highlighted at the end of Section \ref{sec:techniques}, namely, the notion of assigning non-unique strings to indices in $\{1,\dotsc,n\}$, the initial decoding step of identifying strings associated with non-zero entries, or the method of controlling for collisions via erasure-correction coding.
}

\section{Bit Mixing Coding: Test Design and Decoding} \label{sec:bmc}

In this section, we provide the details of BMC, as well as formally stating the guarantees on the number of tests and decoding time.   The main subsequent notation is shown in Table \ref{tbl:notation}

\begin{table}
    \centering
    \caption{Notation used throughout the paper. \label{tbl:notation}}
    \begin{tabular}{|c|l|}
        \hline
        $n$ & Number of items \\
        \hline
        $k$ & Maximum number of defective items (known to the algorithm) \\
        \hline
        $k'$ & Actual number of defective items (not known to the algorithm) \\
        \hline
        $t$ & Total number of tests \\
        \hline
        $\Kc$ & Defective set \\
        \hline
        $\Kchat$ & Estimate of the defective set decoded in the second batch \\
        \hline
        $w$ & Weight of a masking string / block length of a codeword \\
        \hline
        $\sv,\svtil$ & Masking strings \\
        \hline
        $\delta$ & Parameter controlling the error probability \\
        \hline 
        $t_1$, $t_2$ & Number of tests in the first and second batches \\
        \hline
        $\Sc$ & Low collision set \\
        \hline
        $\Lc$ & Set of masking strings decoded in first batch \\
        \hline
        $\Cc$ & Codebook (possibly non-binary) with block length $w$ \\
        \hline
        $\Ac$ & Symbol alphabet for the codebook $\Cc$ \\
        \hline
        $\ell$ & Number of bits to represent a symbol in $\Ac$, i.e., $\ell = \log_2|\Ac|$ \\
        \hline
    \end{tabular}
\end{table}

\subsection{Overview of Bit Mixing Coding} \label{sec:overview}

Here we provide a brief overview of our test design and decoding strategy.  Given integers $t_1$, $t_2$, and $w$, the testing is done in two batches, described below (we use the terminology {\em batches} instead of {\em stages} to highlight that the testing remains entirely non-adaptive).  A rough illustration of these batches is shown in Figure \ref{fig:bmc1}.  Subsequently, the function $\log(\cdot)$ has base $e$.

In the first batch, each item is assigned a binary string of length $t_1$ and weight $w$, chosen uniformly at random with replacement from a carefully designed set $\Sc \subseteq \{0,1\}^{t_1}$.  We refer to these strings as {\em masking strings} (see Section \ref{sec:masking}).  The number of strings in $\Sc$ is typically much smaller than the number of items, implying that a given item's string is unlikely to be unique.  However, we do seek uniqueness {\em among the defective items}.

The testing sub-matrix $\Xv_1 \in \{0,1\}^{t_1 \times n}$ simply arranges the items' strings in columns (or rows in Figure \ref{fig:bmc1}, which shows $\Xv^T$).  Given the resulting $t_1$ test outcomes, the decoder searches through the strings in $\Sc$ and seeks to determine which ones were assigned to {\em some} defective item, but without attempting to identify the index of that item.

In the second batch, the testing sub-matrix $\Xv_2 \in \{0,1\}^{t_2 \times n}$ has a similar structure to $\Xv_1$, but with each bit replaced by a constant number $\ell$ of bits; hence, $t_2 = \ell t_1$.  Any entry that was zero in $\Xv_1$ is simply replaced by a string of $\ell$ zeros.  On the other hand, for any given column, each of the $w$ entries equal to one is replaced by the binary description of a symbol from a codeword.  Specifically, each item has a {\em unique} codeword of length $w$ on an alphabet $\Ac$ of size $2^{\ell}$, and that codeword is an erasure-coded representation of the item's index.

\begin{figure}
  \centering
  \includegraphics[width=0.95\columnwidth]{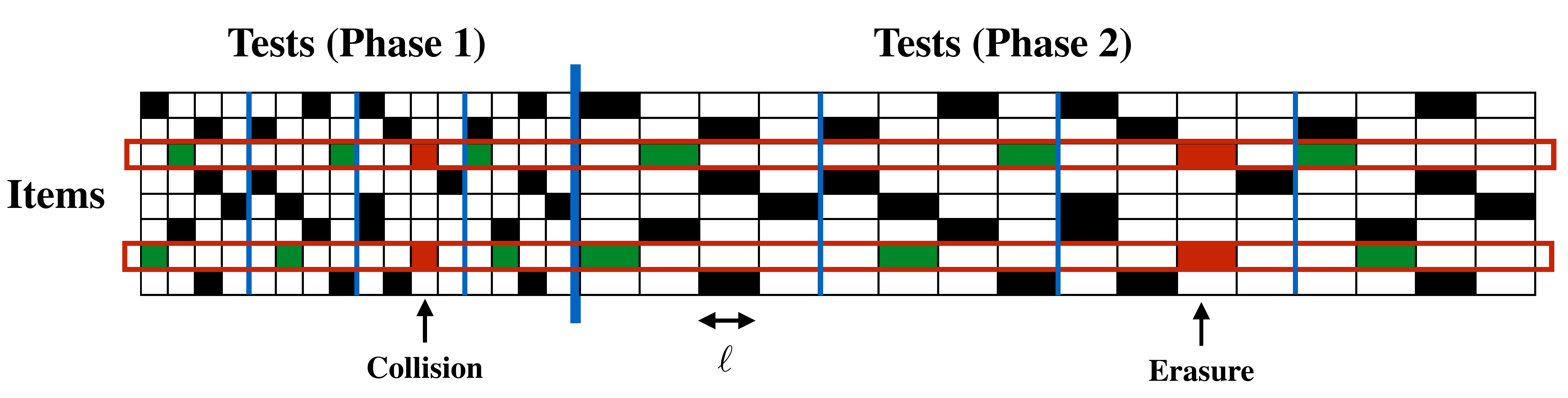}
  \caption{Illustration of the BMC-based (transpose of the) group testing matrix. In the first batch, each item is assigned a constant-weight masking string with weight $w = 4$, and in the second batch, the same structure is repeated in length-$\ell$ segments corresponding to symbols on a larger alphabet. {For compactness, we consider only $n = 8$ items and $k = 2$ defectives, and each masking string has length $t_1 = 16$ even though the choice in our mathematical analysis would correspond to $t_1 = 4kw = 32$.}  The $k=2$ rows corresponding to defective items are highlighted, and we observe that their masking strings collide in the third segment of length $4$. }
  \label{fig:bmc1}
\end{figure}

The idea of the decoding procedure is as follows.  Suppose that we have designed $\Sc$ such that with high probability, (i) the first batch of tests allows the decoder to successfully identify which $k$ (or fewer) masking strings were assigned to defective items; and (ii) any one of these strings collides (i.e., overlaps in the indices equaling $1$) with the union of the $k-1$ other strings in at most $\frac{w}{2}$ indices.\footnote{We will introduce these as key properties of {\em low collision sets} in Section \ref{sec:masking}.}  These properties ensure that from the second batch of tests, the decoder can perfectly recover the symbols (with values in $\Ac$) corresponding to the $\frac{w}{2}$ (or more) non-colliding locations of $1$'s in each defective item's masking string, while marking the symbols in the other $\frac{w}{2}$ (or fewer) locations as erasures.  Any length-$w$ code on $\Ac$ capable of correcting the worst-case erasure of half the codeword symbols can therefore recover this defective item's codeword, and hence also the index of the item.

In the following, we focus on the case that $k \to \infty$ as $n \to \infty$.  The case $k = O(1)$ is in fact much simpler, but also more convenient to handle separately, so it is deferred to Appendix \ref{sec:fixed_k}.

\subsection{Masking Strings and Low Collision Sets} \label{sec:masking}

A key technical challenge in our analysis is proving the existence of the set $\Sc \subseteq \{0,1\}^{t_1}$ satisfying the properties overviewed in Section \ref{sec:overview}.  We proceed by presenting the relevant definitions and results towards achieving this goal.

We begin with the formal definition of a masking string.  This definition depends on the maximum number of defectives $k$ and a length parameter $w$, and leads to a number of tests in the first batch given by $t_1 = 4kw$.

\begin{defn} \label{def:masking}
    We say that $\sv \in \{0,1\}^{t_1}$ is a {\em $(k,w)$ masking string} if it is the concatenation of $w$ (typically different) binary substrings of length $4k$, with each substring having a Hamming weight of $1$.
\end{defn}

We use the simplified terminology {\em masking string} when the parameters $k$ and $w$ are clear from the context.  Clearly, any $(k,w)$ masking string has length $t_1 = 4kw$ and weight $w$.

Our group testing design will rely crucially on a subset $\Sc$ of masking strings that are sufficiently ``well-separated on average''.  Specifically, when we assign masking strings from $\Sc$ to items uniformly at random with replacement, we seek to ensure that (i) upon observing the bitwise ``OR'' of the $k' \le k$ masking strings assigned to defective items, the decoder can identify the corresponding $k'$ (or fewer) individual strings in $\Sc$; and (ii) each of these $k'$ masking strings has at most half of its $1$'s in common with the union of the other $k'-1$.  The following definition formally introduces sufficient requirements for this purpose.

\begin{defn}
    \label{def:lcs}
    A set $\Sc \subseteq \{0,1\}^{t_1}$ of $(k,w)$ masking strings is a {\em $(k,w,\delta)$ low collision set} (LCS) if it satisfies the following property for any given integer $k' \le k$ and any given index $i \in \{1,\dotsc,k'\}$: If we choose $k'$ strings $\svtil_1,\dotsc,\svtil_{k'}$ from $\Sc$ uniformly at random with replacement, then the following conditions hold with probability at least $1-\delta$:
    \begin{enumerate}
        \item The multi-set $\Sctil = \{\svtil_1,\ldots, \svtil_{k'}\}$ is such that all $\svtil \in \Sc \setminus \Sctil$ satisfy $\sum_{j=1}^{k'} \svtil^T \svtil_j \le \frac{w}{2}$;
        \item The multi-set $\Sctil^{(-i)} = \{\svtil_1, \ldots, \svtil_{i-1}, \svtil_{i+1}, \ldots, \svtil_{k'}\}$ is such that $\sum_{1 \le j \le k' \,:\, j \ne i} \svtil_i^T \svtil_j \le \frac{w}{2}$.
    \end{enumerate}
\end{defn}

The bulk of our technical analysis is devoted to proving the following lemma, establishing the existence of an LCS with certain requirements on the size $|\Sc|$ and parameters $(k,w,\delta)$.  To simplify the analysis, we state the result in an asymptotic form, but non-asymptotic variants can easily be deduced from the proof.  In addition, we make no effort to optimize the constant factors, which could also be improved by refining our analysis.

\begin{lem} \label{lem:lcs}
    Consider any sequence of $(k,\delta)$ pairs such that $k \to \infty$, $\delta \to 0$, and $\delta \ge \frac{1}{k^2}$.  If $w$ satisfies
    \begin{equation}
        w \ge 70\, \log \frac{k}{\delta}, \label{eq:choice_w}
    \end{equation}
    then for sufficiently large $k$ there exists a $(k,w,\delta)$ low collision set (LCS) $\Sc$ with cardinality $|\Sc| = \frac{2k}{\delta}$.
\end{lem}
\begin{proof}
    See Section \ref{sec:lcs}.
\end{proof}

While the construction used to prove Lemma \ref{lem:lcs} is randomized, the proof provides sufficient conditions for being an LCS that hold with high probability, {and that can be verified in time $O\big( |\Sc|^2 \cdot w \big)$.  We will later set $w = O(\log n)$ and $\delta = \frac{1}{k \log k}$, in which case substituting $|\Sc| = \frac{2k}{\delta}$ gives verification time $O(k^4 (\log k)^4 \log n) = \poly(k \log n)$.} In contrast, given a set $\Sc$ of masking strings, it appears to be difficult to directly verify whether the set is an LCS in an efficient manner.

\subsection{Encoding and Decoding: First Batch of Tests} \label{sec:batch1}

\begin{algorithm}
    \caption{\label{alg:masking}
    Test design (encoding) and masking string identification (decoding) for the first batch of tests. }
    \small
    \begin{algorithmic}[1]
        \Statex \hspace*{-6mm} {\em Global parameters:} Number of items $n$, triplet $(k,w,\delta)$, low collision set $\Sc$ of size $\frac{2k}{\delta}$
        
        \Statex

        \Statex \hspace*{-6mm} {\em Test design}
        \State {\bf foreach} $i = 1,\dotsc,n$ {\bf do}
        \State \hspace*{4mm} Let the $i$-th column of $\Xv_1$ be a uniformly random element from the set $\Sc$
        \State {\bf endfor}
        
        \Statex

        \setcounter{ALG@line}{0}
        \Statex \hspace*{-6mm} {\em Masking string identification}
        ({\bf input:} A received binary string $\yv_1$ of $4kw$ bits;
        {\bf output:} A list $\Lc$ of masking strings)
        \State {\bf foreach} $\sv \in \Sc$ {\bf do}
        \State \hspace*{4mm} {\bf if} {$\sv^T \yv_1 = w$} {\bf then} include $\sv$ in the output list $\Lc$ \label{line:decision}
        \State {\bf endfor}
    \end{algorithmic}
    \normalsize
\end{algorithm}

The test design and decoding procedure associated with the first batch of tests are depicted in Algorithm \ref{alg:masking}.  The test design simply assigns a masking string to each item uniformly at random from $\Sc$ with replacement, and arranges these in columns to form $\Xv_1$.  Given the resulting test outcome vector $\yv_1 \in \{0,1\}^{t_1}$, the decoder constructs a list $\Lc \subseteq \Sc$ of masking strings believed to correspond to defective items by adding only the strings having sufficient overlap with $\yv_1$ in the locations of $1$'s.

The following lemma provides a formal statement of successful masking string identification.

\begin{lem} \label{lem:batch1}
    Suppose that there are $k' \le k$ defective items, and their associated masking strings $\{\svtil_1,\dotsc,\svtil_{k'}\}$ satisfy the first condition of Definition \ref{def:lcs}.  Then, the test design and masking string identification procedure in Algorithm \ref{alg:masking} lead to an estimate $\Lc$ containing $\{\svtil_1,\dotsc,\svtil_{k'}\}$ and no other elements of $\Sc$.
\end{lem}
\begin{proof}
    It is trivial that any masking string $\svtil_i$ assigned to a defective item will be included in $\Lc$:  Any index where its masking string is $1$ will lead to a positive test, yielding $\svtil_i^T \yv_1 = w$ since the weight of each masking string is $w$.  

    On the other hand, if $\svtil \in \Sc$ is not assigned to any defective item, then the first property of Definition \ref{def:lcs} ensures that the sum of overlaps between $\svtil$ and the elements of $\{\svtil_1,\dotsc,\svtil_{k'}\}$ is at most $\frac{w}{2}$.  Since $\yv_1$ is the bit-wise ``OR'' of $\{\svtil_1,\dotsc,\svtil_{k'}\}$, this implies that $\svtil^T \yv_1 \le \frac{w}{2}$. 
\end{proof}

\subsection{Encoding and Decoding: Second Batch of Tests} \label{sec:batch2}

\begin{algorithm}
    \caption{\label{alg:data}
    Test design (encoding) and item identification (decoding) for the second batch of tests.}
    \small
    \begin{algorithmic}[1]
        \Statex \hspace*{-6mm} {\em Global parameters:} Number of items $n$, triplet $(k,w,\delta)$, multi-set $\{\sv_1,\dotsc,\sv_n\}$ of masking strings assigned to items in first batch, parameter $\ell$ and alphabet $\Ac$ of size $2^{\ell}$, codebook $\Cc = \{\cv_1,\dotsc,\cv_n\}$ with $n$ codewords in $\Ac^w$.
        \Statex
        \Statex \hspace*{-6mm} {\em Test design}
        \State {\bf foreach} $i = 1,\dotsc,n$ {\bf do}
        \State \hspace*{4mm} Initialize $\xv$ to be the empty string
        \State \hspace*{4mm} Let $\sv = (s_1,\dotsc,s_{4kw})$ be the masking string assigned to item $i$ in the first batch
        \State \hspace*{4mm} {\bf foreach} $j = 1,\dotsc,4kw$ {\bf do}
        \State \hspace*{4mm} \hspace*{4mm} {\bf if} $s_j = 0$ {\bf then} append $\ell$ zeros to $\xv$
        \State \hspace*{4mm} \hspace*{4mm} {\bf else} Append length-$\ell$ binary representation of the next symbol of $\cv_i$ to $\xv$
        \State \hspace*{4mm} {\bf endfor}
        \State \hspace*{4mm} Fill in the $i$-th column of $\Xv_2$ with the entries of $\xv$
        \State {\bf endfor}

        \Statex

        \setcounter{ALG@line}{0}
        \Statex \hspace*{-6mm} {\em Item identification}
        ({\bf input:} Received string $\yv_2$ of length $4kw\ell$, list $\Lc$ of decoded masking strings returned by Algorithm~\ref{alg:masking};
        {\bf output:} Estimate $\Kchat$ of the defective set)
        \State Construct $\yvtil \in \Ac^{4kw}$ by converting $\yv_2 \in \{0,1\}^{4kw\ell}$ from binary to the alphabet $\Ac$
        \State {\bf foreach} $\sv \in \Lc$ {\bf do}
        \State \hspace*{4mm} Initialize $\uv$ to be the empty string
        \State \hspace*{4mm} {\bf foreach} $i=1,\dotsc,4kw$ {\bf do}
        \State \hspace*{4mm}\hspace*{4mm} {\bf if} ($s_i = 1$) and
         (there exists no $\svtil \in \Lc$ such that $\svtil \ne \sv$ and $\stil_i = 1$) {\bf then}
        \State \hspace*{4mm}\hspace*{4mm}\hspace*{4mm}
        Append the $i$-th symbol of $\yvtil$ to $\uv$;
        \State \hspace*{4mm}\hspace*{4mm} {\bf else if} ($s_i = 1$) {\bf then}
        \State \hspace*{4mm}\hspace*{4mm}\hspace*{4mm}
        Append the erasure symbol to $\uv$;
        \State \hspace*{4mm} {\bf endfor}
        \State \hspace*{4mm} $\Ic \leftarrow $ decoder for $\Cc$ applied to $\uv$ to return an index in $\{1,\dotsc,n\}$
        \State \hspace*{4mm} Include $\Ic$ in the output set $\Kchat$
        \State {\bf endfor}
    \end{algorithmic}
    \normalsize
\end{algorithm}

The test design and decoding procedure associated with the second batch of tests are depicted in Algorithm \ref{alg:data}.  As discussed in Section \ref{sec:overview}, the idea is to copy the structure of $\Xv_1$, but replace each bit by a sequence of $\ell$ bits.  Any ``$0$'' bit is trivially mapped to a string of $\ell$ zeros, whereas any ``$1$'' bit is replaced by the binary representation of a codeword symbol.  The codeword has length $w$ and alphabet $\Ac$, whose size is $|\Ac| = 2^{\ell}$, and the corresponding codebook $\Cc = \{\cv_1,\dotsc,\cv_n\}$ is chosen to have good worst-case erasure correction guarantees (see Section \ref{sec:erasure}).  For item $i$, the codeword $\cv_i$ is used.

For item identification, any collisions between masking strings in $\Lc$ (returned from the first batch) are treated as erasures, whereas in the absence of a collision, the corresponding length-$\ell$ binary string from the test outcome vector $\yv_2$ is mapped to a symbol from $\Ac$.  For each $\sv \in \Lc$, if there are sufficiently few erasures, then we can recover the corresponding codeword $\cv_i$ via erasure-correcting decoding, and hence identify the defective item index $i \in \{1,\dotsc,n\}$.

The following lemma formally states the requirements on $\Cc$, along with sufficient conditions under which the decoding succeeds.

\begin{lem} \label{lem:batch2}
    Suppose that there are $k' \le k$ defective items, and their associated masking strings $\{\svtil_1,\dotsc,\svtil_{k'}\}$ satisfy the second condition of Definition \ref{def:lcs}.  If the first batch successfully produces $\Lc = \{\svtil_1,\dotsc,\svtil_{k'}\}$, and the decoder of $\Cc$ is able to correct an arbitrary pattern of $\frac{w}{2}$ erasures, then the test design and item identification procedure in Algorithm \ref{alg:data} lead to successful recovery, i.e., $\Kchat = \Kc$.
\end{lem}
\begin{proof}
    The second condition of Definition \ref{def:lcs} implies that $\Lc = \{\svtil_1,\dotsc,\svtil_{k'}\}$ contains no duplicates, and also that for any such $\svtil_i \in \Lc$, at most $\frac{w}{2}$ of the indices of $1$'s collide with those of {\em any} of the other strings in $\Lc$.  Hence, when $\svtil_i$ is processed in the outer loop of item identification in Algorithm \ref{alg:data}, we have the following:
    \begin{itemize}
        \item Whenever there is a collision, an erasure symbol is added to $\uv$, and this occurs at most $\frac{w}{2}$ times;
        \item Whenever there is no collision, the correct codeword symbol from $\cv_i$ is added to $\uv$.
    \end{itemize}
    Hence, $\uv$ equals the desired length-$w$ codeword $\cv_i$ with at most $\frac{w}{2}$ entries replaced by the erasure symbol, and by our assumption on the decoder of $\Cc$, the correct codeword $\cv_i$ (or equivalently, the correct index $i$) is identified.
\end{proof}

\subsection{Choice of Erasure-Correcting Code} \label{sec:erasure}

The problem of decoding in the presence of worst-case erasures has been extensively studied in coding theory.  There are many erasure-correcting codes that we could use in Algorithm \ref{alg:data}, with various trade-offs in the subsequent mathematical analysis and decoding time.  For instance:
\begin{itemize}
    \item In a preliminary version of this work \cite{Bon19}, we used Reed-Solomon codes, which have the convenient feature of being maximum Maximum Distance Separable (MDS).  However, when applied to group testing, their large alphabet size (i.e., increasing in the block length) leads to an $O(k \cdot \log k \cdot \log\log k)$ term in the number of tests.\footnote{In \cite{Bon19} the logarithmic factors were also not optimized, so the analysis therein actually would actually lead to an $O(k \cdot \log^2 k \cdot \log\log k)$ term.}  This suffices for attaining the optimal $t = O(k \log n)$ scaling in sufficiently sparse regimes, but here we prefer to adopt an approach that does so in both sparser and denser regimes.
    \item In Section \ref{sec:noisy}, we discuss the use of binary codes (i.e., $\Ac = \{0,1\}$), which makes Algorithm \ref{alg:data} conceptually simpler, and can be useful in noisy scenarios.  However, this requires several of the constants to be modified to less favorable values throughout the analysis.  For instance, the length $t_1 = 4kw$ may be increased to a value such as $10kw$, and the proportion of erasures permitted may decrease from $\frac{1}{2}$ to a smaller value such as $\frac{1}{10}$.
    \item We ideally seek linear decoding time in the block length, though polynomial decoding time is also acceptable given that the code length is only $w = O(\log n)$.
\end{itemize}
As a suitable trade-off of these various aspects, we found the following code construction from \cite{Alo96} to be convenient, providing near-MDS erasure correction with a {\em bounded} alphabet size.  The code construction is based on expanders.

\begin{lem} \label{lem:erasure}
    {\em \cite[Thm.~1]{Alo96}}
    For any $r \in (0,1)$ and arbitrarily small $\epsilon > 0$, there exists an alphabet $\Ac$ whose size is a constant depending only on $\epsilon$, and a codebook $\Cc$ (with codeword symbols on $\Ac$) and associated encoder/decoder pair, such that the following properties hold:
    \begin{itemize}
        \item $\Cc$ has rate $r$, i.e., the number of codewords is $|\Ac|^{wr}$, where $w$ is the block length;
        \item The decoder corrects any (worst-case) fraction $1-r-\epsilon$ of erasures;
        \item The encoding and decoding time are linear in the block length.
    \end{itemize}
\end{lem}

In our analysis, we will not require $\epsilon$ to be arbitrarily small, and instead simply take $r = \frac{1}{3}$ and $\epsilon = \frac{1}{6}$, so that a fraction $\frac{1}{2}$ of erasures is tolerated.  

\subsection{Statement of Main Result} \label{sec:main}

We are now ready to state our main theorem.  For simplicity, we set the relevant parameters to ensure $\pe \le \frac{1}{\log k}$, but with simple modifications to the constant factors (here and in the auxiliary results), we can improve this to $\pe \le \frac{1}{k^c}$ for any fixed $c > 0$.  However, it is worth noting that the decoding time has a $\frac{1}{\pe}$ dependence on $\pe$, which is why we choose a logarithmic dependence on $k$.  We also re-iterate that we have made no effort to optimize constant factors, and we recall that despite the assumption $k \to \infty$ here, the case $k = O(1)$ is in fact much simpler, and is handled in Appendix \ref{sec:fixed_k}.

\begin{thm}
    \label{the:main}
    Under the choices $w = \max\big\{ \frac{3}{\ell} \log_2 n, 70 \log \frac{k}{\delta}  \big\}$ and $\delta = \frac{1}{k \log k}$, and a code $\Cc$ chosen suitably according to Lemma \ref{lem:erasure}, the BMC group testing procedure described in Algorithms \ref{alg:masking} and \ref{alg:data} with an LCS constructed according to Lemma \ref{lem:lcs} yields $\pe \le \frac{1}{\log k}$ for any $k \to \infty$, and the resulting number of tests used satisfies
    \begin{equation}
        t \le 12 k \max\bigg\{ \frac{\ell+1}{\ell} \log_2 n, ~ 50\, (\ell+1) \log k \bigg\} (1+o(1)), \label{eq:t_final}
    \end{equation}
    where $\ell \ge 2$ is a constant (not depending on $n$ or $k$) corresponding to $\log_2|\Ac|$ in Lemma \ref{lem:erasure}.  In addition, with probability at least $1 - \frac{1}{\log k}$, the decoding time is $O(k^2 \cdot \log k \cdot \log n)$.
\end{thm}
\begin{proof}
    By Lemma \ref{lem:lcs}, there exists an LCS with $\delta = \frac{1}{k \log k}$ as long as $w \ge 70\, \log \frac{k}{\delta}$.  The choice of $w$ in the theorem statement ensures that this condition is true.  Then, since the high-probability event in Definition \ref{def:lcs} holds with probability at least $1-\delta$ for each defective item $i$, it holds simultaneously for all defective items with probability at least $1-k\delta = 1-\frac{1}{\log k}$.  Under this high-probability event, assuming the codebook $\Cc$ in Algorithm \ref{alg:data} corrects $\frac{w}{2}$ worst-case erasures, we deduce from Lemmas \ref{lem:batch1} and \ref{lem:batch2} that the final estimate of the defective set is indeed correct.

    It remains to choose the parameters to ensure that $w \ge 70\, \log \frac{k}{\delta}$, and to characterize the total number of tests and runtime.  Suppose that, as stated following Lemma \ref{lem:erasure}, we use a code of rate $\frac{1}{3}$.  Since identifying an item requires $\frac{1}{\ell} \log_2 n$ symbols from $\Ac$ (with alphabet size $2^{\ell}$), a rate-$\frac{1}{3}$ code yields $w = \frac{3}{\ell} \log_2 n$.\footnote{Here and subsequently, we ignore rounding issues, as these do not impact the final result.}  This is consistent with the condition $w \ge 70\, \log \frac{k}{\delta}$ whenever $k \le \delta  n^{\frac{3}{70 \ell \cdot \log 2}}$.  Alternatively, if this condition on $k$ fails to hold, we can simply replace the rate-$\frac{1}{3}$ code by a (potentially much) lower rate code such that $w = 70 \log \frac{k}{\delta}$; by Lemma \ref{lem:erasure}, such a code still exists with the required erasure-correcting properties.  Combining these two cases, we obtain
    \begin{equation}
        w = \max\bigg\{ \frac{3}{\ell} \log_2 n, 70 \log \frac{k}{\delta}  \bigg\}.
    \end{equation}
    The number of tests is equal to $t_1 = 4kw$ in the first batch, and $t_2 = 4kw\ell$ in the second batch, yielding a total number of tests equal to
    \begin{align}
        t &= 4 kw(\ell + 1) \\
          &= 4 k \max\bigg\{ \frac{3(\ell + 1)}{\ell} \log_2 n, 70 (\ell + 1) \log \frac{k}{\delta}  \bigg\}
    \end{align}
    Substituting $\delta = \frac{1}{k \log k}$, taking a factor of $3$ out the front, and writing $\frac{2 \times 70}{3} \le 50$, we obtain \eqref{eq:t_final}.

    {\bf Decoding time.} For decoding in Algorithm 1, we need to compute an inner product between $\yv_1$ and every $\sv \in \Sc$. To do so, we use the $w$ positions of the ``1'' bits in $\sv$ to index the required entries of $\yv_1$. This leads to $O(w) = O( \log n )$ complexity for each $\sv$, or $O(|\Sc| \cdot \log n) = O(k^2 \cdot \log k \cdot \log n)$ for all $\sv \in \Sc$ (since $|\Sc| = \frac{2k}{\delta}$). The decoding in Algorithm 2 has a total of $|\Lc|$ iterations. In each iteration, it constructs a sequence $\uv$ while incurring $O(w|\Lc|) = O(|\Lc| \log n)$  complexity,\footnote{The loop from $i=1,\dotsc,4kw$ need not be done explicitly; instead, this can be thought of as a loop over $w$ locations of $1$'s.} and then invokes decoding on $\uv$, whose time is linear in the length $w = O(\log n)$ (see Lemma \ref{lem:erasure}).  Hence, the total decoding time for Algorithm 2 is $O( |\Lc|^2 \log n )$.  By Lemma \ref{lem:batch1} and our choice of $\delta$, we know that with probability at least $1-\frac{1}{\log k}$, we have $|\Lc|\le k$.  It is then easy to see that the decoding time is dominated by Algorithm \ref{alg:masking}, and the overall complexity is $O(k^2 \cdot \log k \cdot \log n)$.
\end{proof}

\subsection{Limitations of BMC} \label{sec:limitations}

{The most immediate limitation of BMC is that it has a higher decoding time than certain existing algorithms (notably including GROTEQUE \cite{Cai13} and SAFFRON \cite{Lee15a}) by a factor of $k$.  Hence, it remains an open problem as to whether one can further reduce the decoding time while still maintaining $t = O(k \log n)$. 

Another important limitation is the dependence on the error probability.  We focused on the goal of attaining asymptotically vanishing error probability, and accordingly only targeted $\pe \le \frac{1}{\log k}$ in our main result with $k \to \infty$.  However, in finite-size systems, the speed of convergence to zero can be important, and faster convergence such as $\pe \le k^{-\tau}$ (with $\tau > 0$) is preferable.  While our algorithm and analysis can be adapted to achieve this stricter requirement, the decoding time increases to $k^{2+\tau} \log n$.  In contrast, SAFFRON and GROTESQUE can attain $\pe \le k^{-\tau}$ while only affecting the constant factors in the runtime.

Finally, we re-iterate that the constants factors in Theorem \ref{the:main} are fairly high, since our focus in this paper is on the scaling laws.
}

\section{Proof of Lemma \ref{lem:lcs} (Finding a Low Collision Set)} \label{sec:lcs}

Algorithm \ref{alg:masking} takes as input an LCS, whose properties play a crucial role in proving our main result, Theorem \ref{the:main}.  In this section, we prove the existence of an LCS under suitable parameters, as stated in Lemma \ref{lem:lcs}.  Specifically, we show that if we construct a multi-set in a certain randomized way, then with probability close to $1$, this multi-set will satisfy some {\em sufficient conditions} for being an LCS.  In addition, these sufficient conditions will be verifiable in polynomial time, which is beneficial from a practical point of view.  We emphasize that the LCS is constructed ``offline'' prior to forming the test matrix, and needs to be done only once.

\subsection{A Random Construction} \label{sec:construct}

We will analyze a randomized construction of masking strings (see Definition \ref{def:masking}).  To construct a {\em single} masking string of length $t_1 = 4kw$, for each $4k$-bit segment of the string, we set a uniformly random bit in the segment to be ``1'' and all remaining bits to be ``0''.  To construct a multi-set $\Sc$ containing $|\Sc| = \frac{2k}{\delta}$ random masking strings, we simply repeat this procedure independently $\frac{2k}{\delta}$ times.  This means that $\Sc$ may contain duplicates; however, we will later prove that with high probability, there are no duplicates, so that $\Sc$ is a set.

\subsection{Overview of the Proof}

We will show that with probability approaching one (as $k \to \infty$), the multi-set returned by the above construction is an LCS. Despite the simplicity of the construction, the reasoning is rather complex because there are two sources of randomness involved: The construction is random, while the definition of LCS (Definition~\ref{def:lcs}) also involves its own randomness in the form of random selections from $\Sc$.

To decouple these two forms of randomness, we will introduce the concept of a {\em promising set} (see Section~\ref{sec:promisingset}).  In contrast with LCS, the definition of a promising set does not contain any probability terms.  In addition, we will be able to verify deterministically in polynomial time whether a set is a promising set or not {(see Section \ref{sec:promisingset} for details)}, whereas it is unclear how to check (in polynomial time) whether a set is an LCS.

We will then prove the following: (i) With probability approaching one, the multi-set returned by the random construction in Section~\ref{sec:construct} is a promising set (see Lemma~\ref{lem:promisingexist} below); (ii) A promising set must be an LCS (see Lemma~\ref{lem:promising2low} below) --- namely, being a promising set is a {\em sufficient condition} for being an LCS.  We will prove these claims for $w \ge 70\, \log \frac{k}{\delta}$, $k \to \infty$, $\delta \to 0$, and $\delta \ge \frac{1}{k^2}$, as stated in Lemma \ref{lem:lcs}.  In fact, the latter condition can be improved to  $\delta \ge \frac{1}{k^c}$ for any constant $c > 0$, by suitably adjusting certain other constants.

In addition to the assumption $w \ge 70\, \log\frac{k}{\delta}$, we can further restrict our attention to $w = C\, \log\frac{k}{\delta}$ for some constant $C = \Theta(1)$ with $C \ge 70$.  Once this is established, we can easily get an LCS for larger $w$ values (e.g., $C \to \infty$ as $k \to \infty$) by repeating each masking string; this is formally stated as follows.

\begin{lem}
    \label{lem:extend}
    Given any $(k, w, \delta)$ low collision set $\Sc$ and any positive integer $c$, we can construct a $(k, cw, \delta)$ low collision set $\Sc^c$.
\end{lem}
\begin{proof}
    For compactness, throughout this proof we use the terminology that a $(k,w)$ masking string $\svtil$ is {\em $w$-compatible} with a multi-set $\{\sv_1,\dotsc,\sv_m\}$ if $\sum_{i=1}^m \svtil^T \sv_i \le \frac{w}{2}$.

    Let $\Sc^c = \{\sv^c \,\,|\,\, \sv\in \Sc \}$, where $\sv^c$ denotes the concatenation of $c$ copies of $\sv$.  For all $\sv$ and $\svtil$, we trivially have $(\sv^c)^T \svtil^c = c\times (\sv^T \svtil )$.  Fix an integer $m$ and a multi-set $\{\svtil^c_1, \ldots, \svtil^c_m\} \subseteq \Sc^c$.  It is easy to verify that: (i) $\Sctil^c = \{\svtil_1^c, \ldots, \svtil_m^c\}$ is $(cw)$-compatible with all $\sv^c \in \Sc^c\setminus \Sctil^c$ if and only if $\Sctil = \{\svtil_1, \ldots, \svtil_m\}$ is $w$-compatible with all $\sv \in \Sc\setminus \Sctil$, and (ii) for all $i=1,\dotsc,m$, $(\Sctil^{(-i)})^c = \{\svtil_1^c, \ldots, \svtil_{i-1}^c, \svtil_{i+1}^c, \ldots, \svtil_m^c\}$ is $(cw)$-compatible with $\svtil_i^c$ if and only if $\Sctil^{(-i)} = \{\svtil_1, \ldots, \svtil_{i-1}, \svtil_{i+1}, \ldots, \svtil_m\}$ is $w$-compatible with $\svtil_i$.  From Definition \ref{def:lcs}, we deduce that since $\Sc$ is a $(k,w,\delta)$ LCS, $\Sc^c$ is a $(k,cw,\delta)$ LCS.
\end{proof}

Hence, we proceed by assuming that $w = C\, \log\frac{k}{\delta}$ with $C = \Theta(1)$ and $C \ge 70$.  In particular, we will use the fact that $\frac{w}{k} \to 0$, obtained by combining this assumption with $k \to \infty$ and $\delta \ge \frac{1}{k^2}$.

\subsection{The Concept of a Promising Set} \label{sec:promisingset}

Given a set $\Sc$ of masking strings and any $\svtil\in \Sc$, we define
\begin{equation}
    \mu(\svtil, \Sc) = \frac{1}{|\Sc|-1} \sum_{\sv\in \Sc \setminus \{\svtil\}} \svtil^T \sv. \label{eq:def_mu}
\end{equation}
In the following, we define the concept of a {\em promising set}, which will provide a stepping stone to establishing the existence of an LCS.

\begin{defn} \label{def:promising}
    A set $\Sc$ of $(k,w)$ masking strings is a {\em $(k,w,\delta)$ promising set} if the following equations hold for all $\svtil \in \Sc$:
    \begin{align}
        \label{eqn:proof1}
        \Big|\mu(\svtil,\Sc) - \frac{w}{4k}\Big| &\le \frac{0.04w}{4k} \\
        \label{eqn:proof2}
        \max_{\sv \in \Sc \setminus \{\svtil\}} |\svtil^T\sv - \mu(\svtil,\Sc)| &\le 6.1 \\
        \label{eqn:proof3}
        \sum_{\sv \in \Sc \setminus \{\svtil\}} (\svtil^T\sv - \mu(\svtil,\Sc))^2
        &\le (|\Sc|-1) \frac{w}{2k}.
    \end{align}
\end{defn}

To gain some intuition behind this definition, note that $\mu(\svtil,\Sc)$ is the average number of collisions between $\svtil$ and other masking strings in $\Sc$. Hence, \eqref{eqn:proof1} requires the average to be close to $\frac{w}{4k}$.  Similarly, \eqref{eqn:proof2} requires the maximum number of collisions to be close to this average, and \eqref{eqn:proof3} bounds the ``variance'' of the number of collisions between $\svtil$ and other masking strings in $\Sc$. The values on the right-hand side of the three equations are carefully chosen
such that (i) the random construction in Section \ref{sec:construct} returns a promising set with high probability, and (ii) a promising set must be an LCS.

{As we stated previously, the conditions in Definition \ref{def:promising} can be verified in a computationally efficient manner.  Computing the inner product between two masking strings can be done in $O(w)$ time, since each has only $w$ non-zero entries.  Then, each mean value $\mu(\svtil, \Sc)$ (for $\sv \in \Sc$) can be computed in time $O(|\Sc| \cdot w)$, for a total of $O(|\Sc|^2 \cdot w)$ time.  Finally, once all such values have been computed, conditions \eqref{eqn:proof1}--\eqref{eqn:proof3} can similarly be directly checked for all $\svtil \in \Sc$ in time $O(|\Sc|^2 \cdot w)$.}


\subsection{Probability of Being a Promising Set} \label{sec:claim1}

The following lemma proves that the random construction in Section~\ref{sec:construct} yields a promising set with high probability.

\begin{lem}
    \label{lem:promisingexist}
    Consider any sequence of triplets $(k,w,\delta)$ such that $k \to \infty$, $\delta \to 0$, $\delta \ge \frac{1}{k^2}$, and $w = C \log\frac{k}{\delta}$ with $C = \Theta(1)$ and $C \ge 70$.  For sufficiently large $k$, with probability\footnote{The probability is with respect to the randomness in the construction in Section \ref{sec:construct}.} approaching one as $k \to \infty$ the multi-set $\Sc$ is a $(k, w, \delta)$ promising set of size $\frac{2k}{\delta}$.
\end{lem}
\begin{proof}
    Let $\Sc = \{\sv_1, \sv_2, \ldots, \sv_{\frac{2k}{\delta}}\}$ be the multi-set constructed in Section~\ref{sec:construct}. With a slight abuse of notation, for any $i$, we define
    $\mu(\sv_i, \Sc) = \frac{1}{|\Sc|-1}\sum_{j\,:\,j\ne i} \sv_i^T \sv_j$. We will prove that, with probability approaching one, the following conditions hold simultaneously for all $i$:
    \begin{align}
        \label{eqn:proof4}
        \Big|\mu(\sv_i, \Sc) - \frac{w}{4k}\Big| &\le \frac{0.04w}{4k} \\
        \label{eqn:substitute}
        \max_{j \,:\, j\ne i} \Big|\sv_i^T \sv_j - \frac{w}{4k}\Big| &\le 6.05 \\
        \label{eqn:substitute2}
        \sum_{j\,:\,j\ne i}\Big(\sv_i^T \sv_j-\frac{w}{4k}\Big)^2
        &\le (|\Sc|-1) \frac{w}{2k}.
    \end{align}
    Note that \eqref{eqn:substitute} implies that $\Sc$ is a set (i.e., there are no duplicates): If there existed $i$ and $j$ such that $i\ne j$ and $\sv_i = \sv_j$, then we would have $\max_{j \,:\,  j\ne i} |\sv_i^T \sv_j - \frac{w}{4k}| = w - \frac{w}{4k} = w(1+o(1))$, violating \eqref{eqn:substitute}.
    Given that $\Sc$ is a set, \eqref{eqn:proof4} becomes equivalent to \eqref{eqn:proof1}.
    Then, combining \eqref{eqn:proof4} and \eqref{eqn:substitute} leads to \eqref{eqn:proof2}, since
    \begin{align}
        \max_{\sv\in \Sc \setminus \{\svtil\}} |\svtil^T\sv - \mu(\svtil,\Sc)| 
            &= \max_{j\,:\,j\ne i} |\sv_i^T \sv_j - \mu(\sv_i,\Sc)| \\
            &\le \Big|\mu(\sv_i, \Sc) - \frac{w}{4k}\Big| + \max_{j\,:\,j\ne i} \Big|\sv_i^T \sv_j - \frac{w}{4k}\Big| \\
            & \le \frac{0.04w}{4k} + 6.05 \\
            &\le 6.1, \label{eq:final6}
    \end{align}
    where \eqref{eq:final6} holds for sufficiently large $k$ since $\frac{w}{k} \to 0$.

    Finally, using the standard fact that an expectation $\EE[(Z - a)^2]$ is always smallest when $a = \EE[Z]$, and noting that $\mu(\sv_i,\Sc)$ is the average of the $\sv_i^T \sv_j$ values for $j \ne i$, we deduce that $\sum_{j\,:\,j\ne i}(\sv_i^T \sv_j -\mu(\sv_i, \Sc))^2 \le \sum_{j\,:\,j\ne i}(\sv_i^T \sv_j-a)^2$ for any $a \in \RR$. Taking $a = \frac{w}{4k}$, we find that \eqref{eqn:substitute2} implies \eqref{eqn:proof3}. 
    
    To complete the proof, we show that \eqref{eqn:proof4}, \eqref{eqn:substitute}, and \eqref{eqn:substitute2} each hold (simultaneously for all $i$) with probability approaching one as $k \to \infty$.  A trivial union bound then shows that the three hold simultaneously with probability approaching one.
    
    For \eqref{eqn:proof4},
    consider any fixed $i$ and fixed $\sv_i$, and view the remaining masking strings in $\Sc$ as random variables (according to the randomness in the construction). The quantity $\sum_{j\,:\,j\ne i} \sv_i^T \sv_j$ follows a binomial distribution with parameters $(|\Sc|-1)w$ and $\frac{1}{4k}$. By the Chernoff bound (see Appendix \ref{sec:conc}), we have
    \begin{align}
        \PP\bigg[\Big|\mu(\sv_i, \Sc) - \frac{w}{4k}\Big| \ge \frac{0.04w}{4k}\bigg] 
            &= \PP\bigg[\Big|\sum_{j\,:\,j\ne i} \sv_i^T \sv_j - (|\Sc|-1)\frac{w}{4k}\Big| \ge (|\Sc|-1)\frac{0.04w}{4k}\bigg] \label{eq:cond1_1} \\
            &\le 2 \exp\bigg(-\frac{1}{3} \cdot (0.04^2) \cdot (|\Sc|-1)\frac{w}{4k}\bigg)   \\
            &\le 2 \exp\bigg(- \frac{70\cdot(0.04^2)}{12 \delta} \log\frac{k}{\delta}\bigg)  \label{eq:cond1_2} \\
            &= \Big(\frac{\delta}{k}\Big)^{\omega(1)},  \label{eq:cond1_3}
    \end{align}
    where \eqref{eq:cond1_2} uses $|\Sc|-1 = \frac{2k}{\delta} - 1 \ge \frac{k}{\delta}$ and $w \ge 70 \log \frac{k}{\delta}$, and \eqref{eq:cond1_3} uses $\delta \to 0$.  By a union bound across all $\frac{2k}{\delta}$ values of $i$, we deduce that \eqref{eqn:proof4} holds with probability approaching one.

    For \eqref{eqn:substitute}, first observe that trivially $\sv_i^T \sv_j \ge \frac{w}{4k} - 6.05$ for sufficiently large $k$, since $\frac{w}{k} \to 0$ and $\sv_i^T \sv_j \ge 0$.  To establish the other direction  $\sv_i^T \sv_j \le \frac{w}{4k} + 6.05$, consider any fixed $i$ and fixed $\sv_i$, and view $\sv_j$ as a random variable.
    The quantity $\sv_i^T \sv_j$ follows a binomial distribution with parameters $w$ and $\frac{1}{4k}$, so its mean is $\frac{w}{4k}$.  By the Chernoff bound (see Appendix \ref{sec:conc}), we have for any $\eta > 0$ that
    \begin{align}
        \PP\Big[\sv_i^T \sv_j \ge \frac{w}{4k} (1+\eta)\Big]
            &\le \exp\bigg( - \frac{w}{4k} \big( (1+\eta) \log(1+\eta) - \eta \big) \bigg). \label{eqn:proof7}
    \end{align} 
    Recalling that $w = C\log\frac{k}{\delta}$ with $C = \Theta(1)$, we set $\eta = \frac{24.2}{C}\cdot \frac{k}{\log \frac{k}{\delta}}$, so that the event in the probability \eqref{eqn:proof7} is indeed the complement of the event $\sv_i^T \sv_j < \frac{w}{4k} + 6.05$.  This choice satisfies $\eta \to \infty$, and hence $ (1+\eta) \log(1+\eta) - \eta = (\eta \log \eta)(1+o(1))$.  Also noting that $\log \eta = (\log k)(1+o(1))$, we find that \eqref{eqn:proof7} simplifies to
    \begin{align}
        \PP\Big[\sv_i^T \sv_j \ge \frac{w}{4k} (1+\eta)\Big]
            &\le \exp\bigg( - \frac{C \log \frac{k}{\delta}}{4k} \cdot \frac{24.2}{C}\cdot \frac{k}{\log \frac{k}{\delta}} \cdot (\log k)(1+o(1)) \bigg) \\
            &= \exp\Big( - \big( 6.05 \log k \big) (1+o(1)) \Big) \\
            &= k^{-6.05(1+o(1))}. \label{eq:cond2_end}
    \end{align}
    There are a total of ${|\Sc| \choose 2} \le (\frac{2k}{\delta})^2$ possible combinations of $i$ and $j$, which we can further upper bound by $4k^6$ since $\delta \ge \frac{1}{k^2}$.  Taking a union bound over all such combinations, we deduce that \eqref{eqn:substitute} holds for all $i$ with probability approaching one.
    
    Finally, for \eqref{eqn:substitute2}, consider any fixed $i$ and $\sv_i$, and view the remaining masking strings in $\Sc$ as random variables. Under the given $i$ and $\sv_i$, define the random variable $Z_j=
    \frac{(\sv_i^T \sv_j-\frac{w}{4k})^2}{(6.05^2)}$ for $j\ne i$. The quantity $\sv_i^T \sv_j$ is a binomial random variable with parameters $w$ and $\frac{1}{4k}$, and hence
    \begin{align}
        \EE[Z_j] &= \frac{\EE[(\sv_i^T \sv_j-\frac{w}{4k})^2]}{(6.05^2)}
            = \frac{\mbox{Var}[\sv_i^T \sv_j]}{(6.05^2)} 
            = \frac{w\cdot \frac{1}{4k} \cdot (1-\frac{1}{4k})}{(6.05^2)},
    \end{align}
    which implies
    \begin{equation}
        \EE\Big[\sum_{j\,:\,j\ne i}Z_j\Big] \le \frac{w(1-\frac{1}{4k})}{4k} \cdot \frac{|\Sc|-1}{(6.05^2)}. \label{eq:mean_sum_Z}
    \end{equation}
    By \eqref{eq:cond2_end} and the union bound, we know that with probability at least $1-k^{-3.05(1+o(1))}$, it holds that $|\sv_i^T \sv_j -\frac{w}{4k}|< 6.05$ for all $j \ne i$, and hence $Z_j \le 1$.  It will be useful to condition on the corresponding event $\Bc = \bigcap_{j\,:\,j \ne i} \{Z_j \le 1\}$.  Since $\{Z_j\}$ are independent random variables, they remain independent after this conditioning.  {In addition, \eqref{eq:mean_sum_Z} implies that
    \begin{equation}
        \EE\Big[\sum_{j\,:\,j\ne i}Z_j \,\Big|\, \Bc\Big] \le \frac{w(1-\frac{1}{4k})}{4k} \cdot \frac{|\Sc|-1}{(6.05^2)}, \label{eq:mean_sum_Z2}
    \end{equation}
    since the conditioning on $\Bc$ does not increase the average (we are conditioning on each $Z_j$ taking smaller values compared to its full range).}

    Conditioned on $\Bc$, we invoke the Chernoff bound (see Appendix \ref{sec:conc}) and get 
    \begin{align}
        &\PP\Big[\sum_{j\,:\,j\ne i}\Big(\sv_i^T \sv_j-\frac{w}{4k}\Big)^2 \ge (|\Sc|-1) \frac{w}{2k} \,\Big|\, \Bc \Big] \\
            &\qquad= \PP\Big[\sum_{j\,:\,j\ne i} Z_j \ge \frac{|\Sc|-1}{(6.05^2)} \frac{w}{2k}  \,\Big|\, \Bc \Big] \\
            &\qquad\le \exp\Big(-\frac{1}{3} \frac{|\Sc|-1}{(6.05^2)} \frac{w(1-\frac{1}{4k})}{4k}\Big)  \label{eq:cond3_6} \\
            &\qquad\le \exp\Big(-\frac{70}{3 (6.05^2)} \cdot \frac{k(1-\frac{1}{4k})\log \frac{k}{\delta}}{4k\delta} \Big)  \label{eq:cond3_7} \\
            &\qquad= \Big(\frac{\delta}{k}\Big)^{\omega(1)}, \label{eq:cond3_8}
    \end{align}
    where the application of the Chernoff bound in \eqref{eq:cond3_6} also uses \eqref{eq:mean_sum_Z2}, \eqref{eq:cond3_7} uses $|\Sc| -1 = \frac{2k}{\delta} - 1 \ge \frac{k}{\delta}$ and $w \ge 70 \log\frac{k}{\delta}$, and \eqref{eq:cond3_8} uses $\delta = o(1)$.  Using $\PP[\Bc] \ge 1-k^{-3.05(1+o(1))}$ and taking a union bound across all values of $i$, we deduce that \eqref{eqn:substitute2} holds with probability approaching one.
\end{proof}

\subsection{A Promising Set Must Be an LCS} \label{sec:claim2}

The following lemma establishes that any promising set is an LCS.

\begin{lem}
    \label{lem:promising2low}
    Consider any sequence of triplets $(k,w,\delta)$ such that $k \to \infty$, $\delta \to 0$, $\delta \ge \frac{1}{k^2}$, and $w \ge 70 \log\frac{k}{\delta}$.  For sufficiently large $k$, a $(k, w, \delta)$  promising set $\Sc$ of size $\frac{2k}{\delta}$ must be a $(k, w, \delta)$ LCS.
\end{lem}
\begin{proof}
    In accordance with Definition \ref{def:lcs}, fix $k' \le k$, and select $k'$ strings $\svtil_1,\dotsc,\svtil_{k'}$ from $\Sc$ uniformly at random with replacement to form the multi-set $\Sctil = \{\svtil_1, \ldots, \svtil_{k'}\}$.  Note that here $\Sc$ is already fixed --- only $\svtil_1,\dotsc,\svtil_{k'}$ are random variables.
    
    We first prove that $\Sc$ satisfies the first requirement of LCS. Specifically, we will show that with probability at least $1-\frac{\delta}{4}$, we have $\sum_{i=1}^{k'} \svtil^T \svtil_i \le \frac{w}{2}$ for all $\svtil\in \Sc \setminus \Sctil$. We consider a binary matrix whose $|\Sc|^{k'}$ columns correspond to all the possible $\Sctil$, and whose $|\Sc|$ rows correspond to all the possible $\svtil\in \Sc$. We say that a matrix entry corresponding to a given $\Sctil$ and $\svtil$ is {\em bad} if $\sum_{i=1}^{k'} \svtil^T \svtil_i > \frac{w}{2}$ and $\svtil\in \Sc \setminus \Sctil$. To prove the desired claim, it suffices to show that at least $|\Sc|^{k'}\times (1-\frac{\delta}{4})$ columns contain no bad entries. Directly proving this appears to be challenging, so we instead prove that for each row, at most a $\frac{\delta}{4|\Sc|}$ fraction of the entries are bad. This will then imply that the total number of bad entries in the matrix is at most $|\Sc|^{k'} \times |\Sc| \times \frac{\delta}{4|\Sc|} = |\Sc|^{k'} \times \frac{\delta}{4}$, and hence there can be at most $|\Sc|^{k'} \times \frac{\delta}{4}$ columns containing bad entries.
    
    To prove that each row has at most $\frac{\delta}{4 |\Sc|}$ fraction of its entries being bad, it suffices to prove that for any given $\svtil$, when we choose $\svtil_1$ through $\svtil_{k'}$ from $\Sc\setminus \{\svtil\}$ uniformly at random with replacement, we have
    \begin{align}
        \label{eqn:proof8}
        \PP\Big[\sum_{i=1}^{k'} \svtil^T \svtil_i \ge \frac{w}{2}\Big] &\le \frac{\delta}{4|\Sc|}.
    \end{align}
    To prove \eqref{eqn:proof8}, define $Z_i= \svtil^T \svtil_i-\mu(\svtil,\Sc)$ for $i=1,\dotsc,k'$, where $\mu(\svtil,\Sc)$ is defined in \eqref{eq:def_mu}.  Hence, we have $\EE[Z_i]=0$. (Note, however, that $\svtil^T \svtil_i$ does {\em not} follow a binomial distribution.)
    Since $\Sc$ is a promising set, \eqref{eqn:proof1} yields
    \begin{equation}
        \sum_{i=1}^{k'}\svtil^T \svtil_i = \sum_{i=1}^{k'}(Z_i+\mu(\svtil,\Sc)) = k'\cdot \mu(\svtil,\Sc) + \sum_{i=1}^{k'} Z_i \le \frac{1.04w}{4} + \sum_{i=1}^{k'} Z_i. \label{eq:sum_mean}
    \end{equation}
    In addition, for all $i=1,\dotsc,k'$, \eqref{eqn:proof2} and \eqref{eqn:proof3} tell us that
    $|Z_i| \le 6.1$ and $\EE[Z_i^2] \le \frac{w}{2k}$.
    Hence, by Bernstein's inequality (see Appendix \ref{sec:conc}), we have
    \begin{align}
        \PP\Big[\sum_{i=1}^{k'} Z_i>\frac{0.96w}{4}\Big]
            &\le \exp\bigg(-\frac{\frac{(0.96w)^2}{16}}{2k' \cdot \frac{w}{2k} + \frac{2}{3} \cdot 6.1 \cdot \frac{0.96w}{4}}\bigg) \label{eq:pro1_2} \\
            &\le \exp\Big(-\frac{\frac{0.96^2}{16}w}{1 + \frac{2}{3} \cdot 6.1 \cdot \frac{0.96}{4}}\Big) \label{eq:pro1_3}  \\
            &\le \exp\Big(-\frac{w}{35} \Big) \label{eq:pro1_4}  \\
            &\le \Big(\frac{\delta}{k}\Big)^{2}  \label{eq:pro1_5} \\
            &\le \frac{\delta}{4|\Sc|}, \label{eq:pro1_6} 
    \end{align}
    where \eqref{eq:pro1_3} uses $k' \le k$, \eqref{eq:pro1_4} uses a numerical calculation, \eqref{eq:pro1_5} uses $w \ge 70 \log \frac{k}{\delta}$, and \eqref{eq:pro1_6} holds since $|\Sc| = \frac{2k}{\delta}$ and $k \to \infty$.  In turn, for any given $\svtil\in \Sc$, \eqref{eqn:proof8} follows since
    \begin{equation}
        \PP\bigg[\sum_{i=1}^{k'} \svtil^T \svtil_i \ge \frac{w}{2}\bigg] \le \PP\bigg[\sum_{i=1}^{k'} Z_i>\frac{0.96w}{4}\bigg] \le \frac{\delta}{4|\Sc|},
    \end{equation}
    where the first inequality uses \eqref{eq:sum_mean}.
    
    Next, we prove that $\Sc$ satisfies the second requirement for an LCS. Specifically, we show that for any given $i \in \{1,\dotsc,k'\}$, with probability at least $1-0.6\, \delta$, the multi-set $\{\svtil_1, \ldots, \svtil_{i-1}, \svtil_{i+1}, \ldots, \svtil_{k'}\}$ and $\svtil_i$ satisfy $\sum_{j \,:\, j \ne i} \svtil_i^T \svtil_j \le \frac{w}{2}$. We clearly only need to prove this for $k' \ge 2$. In addition, since all the $\svtil_i$'s are generated in a symmetric manner, we can assume without loss of generality that $i = k'$. 

    Define $\Sctil^{(-k')} = \{\svtil_1, \svtil_2, \ldots, \svtil_{k'-1}\}$. We claim that with probability at least $1-0.5\delta$, $\svtil_{k'}\notin \Sctil^{(-k')}$. To see this, note that $\svtil_1,\dotsc,\svtil_{k'-1}$ correspond to at most $k'-1$ distinct elements form $\Sc$, and hence
    $\PP[\svtil_{k'}\in \Sctil^{(-k')}] \le \frac{k-1}{|\Sc|} \le 0.5\delta$.
    Conditioned on $\svtil_{k'}\notin \Sctil^{(-k')}$, each $\svtil_j$ for $1\le j\le k'-1$ is a uniformly random string in $\Sc\setminus \{\svtil_{k'}\}$.  As a result, one can apply the same analysis as that for \eqref{eqn:proof8} (after replacing $k'$ by $k'-1$), and deduce that
    \begin{equation}
        \PP\bigg[\sum_{j=1}^{k'-1} (\svtil_{k'}^T \svtil_j) \ge \frac{w}{2}\bigg] \le \frac{\delta}{4|\Sc|} = o(\delta),
    \end{equation}
    where we used the fact that $|\Sc| = \frac{2k}{\delta} \to \infty$.  Hence, we know that with probability at least $(1-0.5\delta)\cdot (1-o(\delta)) \ge 1-0.6\delta$ (for sufficiently large $k$), the multi-set $\Sctil^{(-k')}$ and $\svtil_{k'}$ satisfy $\sum_{1 \le j \le k'-1} \svtil_{k'}^T \svtil_j \le \frac{w}{2}$.
    
    Finally, a union bound over the two requirements shows that the requirements hold simultaneously with probability at least $1-\delta$, meaning that $\Sc$ is an LCS.
\end{proof}

\section{Extension to the Noisy Setting} \label{sec:noisy}

While the noiseless group testing model is in itself of significant interest, there is also substantial motivation to develop algorithms with low decoding time in the presence of noise.  For combinatorial group testing, it is common to assume a bounded number of {\em worst case} errors (e.g., see \cite{Che09}), whereas for probabilistic group testing, it is more common to assume that tests are subject to {\em random} noise (e.g., see \cite{Cai13,Lee15a,Ina19}).  We focus on the latter, and then briefly discuss the former.

Specifically, we outline a natural extension of BMC (i.e., Algorithms \ref{alg:masking} and \ref{alg:data}) and Theorem \ref{the:main} to the noisy setting.  Generalizing \eqref{eq:gt_noiseless_model}, we consider the following widely-adopted symmetric noise model:
\begin{equation}
    Y = \bigg( \bigvee_{j \in \Kc} X_j \bigg) \oplus Z, \label{eq:gt_symm_model}
\end{equation}
where $Z \sim \Bernoulli(\xi)$ for some constant $\xi \in \big[0,\frac{1}{2}\big)$,  and $\oplus$ denotes modulo-2 addition.  We assume that the noise is independent between tests, i.e., we have i.i.d.~bit flips.

In Sections \ref{sec:bmc} and \ref{sec:lcs}, we used masking strings with length $t_1 = 4kw$, and showed that this leads to at most $\frac{w}{2}$ collisions in each defective item's masking string, with high probability.  In the following, we make use of the following more general statement: For masking strings of length $t_1 = c_1kw$ constructed by concatenating $w$ unit-weight substrings of length $c_1 k$ for some constant $c_1 \ge 4$, we have 
\begin{equation}
    t_1 = c_1 k w \implies \text{At most }\frac{2w}{c_1}\text{ collisions} \label{eq:c1}
\end{equation}
in each defective item's masking string, with high probability.  This follows from straightforward modifications of our previous analysis, including its associated constant factors. 

{For the first batch of tests, we can modify the decision step (Line \ref{line:decision} of the second part of Algorithm \ref{alg:masking})  to the following for improved robustness:
\begin{equation}
    \text{if $\sv^T\yv_1 \ge \frac{3w}{4}$ then include $\sv$ in the output list $\Lc$.} \label{eq:thresh}
\end{equation}}
As seen in the proof of Lemma \ref{lem:batch1} the values of $\sv^T\yv_1$ that we obtain in the absence of noise are exactly $w$ for masking strings of defective items, and at most $\frac{w}{2}$ for the other masking strings.  Hence, as long as fewer than $\frac{w}{4}$ bit flips occur in the entries of $\yv_1$ corresponding to ones in $\svtil$, the correct decision is still made.  

Under the above model of i.i.d.~bit flips, we can simply use the Chernoff bound for an i.i.d.~sum of $w$ random variables (see Appendix \ref{sec:conc}), and deduce that if $\xi < \frac{1}{4}$, then the mis-classification event resulting from Algorithm \ref{alg:masking} has probability $O(n^{-c})$, where $c$ can be set to an arbitrary value by choosing the implied constant in $w = \Theta(\log n)$ large enough.  Choosing $c$ large enough, the error probability remains small even after a union bound over the $\frac{2k}{\delta}$ masking strings.  In the case that $\xi \in \big[\frac{1}{4},\frac{1}{2}\big)$, we can increase the value of $c_1 \ge 4$ and use \eqref{eq:c1}, so that $\svtil^T \yv_1$ reduces from $\frac{w}{2}$ to $\frac{2w}{c_1}$ (or less) for masking strings not assigned to defective items.  Upon changing the threshold from $\frac{3w}{4}$ to $\frac{1}{2}\big( \frac{2w}{c_1} + w \big)$ in \eqref{eq:thresh}, the preceding argument generalizes easily to this case, permitting any noise level $\xi \in \big [0,\frac{1}{2}\big)$ as long as $c_1$ is large enough.

For the second batch of tests, when noise is present, we can no longer assume that the symbols at any non-collided locations are received perfectly.  However, since this part is based on erasure-correcting coding, we can easily generalize to {\em erasure and error correcting coding} to achieve tolerance to noise.  

In the presence of noise, the use of non-binary codes with symbols mapped directly to $\ell > 1$ bits (see Algorithm \ref{alg:data}) may not be ideal, since even a single flip among these $\ell$ bits will cause the symbol to be changed.  We therefore favor the use of a binary code $\Cc$ in the noisy setting, along with a suitable modification of the constants.  
In this case, we again use the more general statement in \eqref{eq:c1} with $c_1 \ge 4$, ensuring at most $\frac{2w}{c_1}$ erasures with high probability.  While a code with minimum distance exceeding $\frac{2w}{c_1}$ would suffice for correcting these erasures alone, here we further increase the target minimum distance beyond $\frac{2w}{c_1}$ in order to account for the bit flips.  

To give a specific example of a binary code with good distance properties, we note that \cite{Gur05} provides a code with linear encoding/decoding time achieving the Blokh-Zyablov bound  \cite[Fig.~1]{Dum98}, with example rate/distance pairs $(R,d_{\min})$ satisfying (i) $R > \frac{1}{5}$ and $d_{\min} > \frac{w}{10}$; (ii) $R > 0.04$ and $d_{\min} > \frac{w}{4}$.  In particular, the rate of the code remains positive as long as $\frac{d_{\min}}{w}$ is a constant strictly less than $\frac{1}{2}$.

To simplify the discussion, suppose that we naively replace all erasures by arbitrary bit values ($0$ or $1$), so that we only have bit flips; this allows us to use the fact that the codes from \cite{Gur05} that permit efficiently decoding any number of worst-case bit flips less than half the minimum distance.  Since the bit flips are i.i.d., we can characterize the number of flips using a concentration argument: With a low enough code rate to make the code length long enough (i.e., a large enough implied constant in $w = O(\log n)$), the number of bit flips is at most $(\xi + \eta) w$ with probability $O(n^{-c})$ for any target $c > 0$, where $\eta > 0$ is any (small) constant.  With at most $(\xi + \eta) w$ bit flips coming from the noise, and at most $\frac{2w}{c_1}$ bit flips coming from the collisions in the first batch, we find that the errors can be corrected as long as $\big( \xi + \eta + \frac{2}{c_1} \big) w < \frac{d_{\min}}{2}$.  Since $d_{\min}$ can be arbitrarily close to $\frac{w}{2}$, this condition can always be satisfied for sufficiently large $c_1$ and a sufficiently low code rate as long as $\xi < \frac{1}{4}$.  
In addition, the case $\xi \in \big[ \frac{1}{4}, \frac{1}{2}\big)$ can be handled similarly as long as one has access to an efficiently decodable constant-rate code that can simultaneously correct $\frac{2w}{c_1}$ worst-case erasures and probability-$\xi$ i.i.d.~bit flips; the condition $\xi < \frac{1}{4}$ above only arose due to using a worst-case error correcting code to correct i.i.d.~bit flips.

In summary, under i.i.d.~noise of the form \eqref{eq:gt_symm_model}, by modifying only the constant factors and the code $\Cc$ used, we can achieve the same scaling laws as Theorem \ref{the:main} in terms of both tests and runtime (at least when $\xi < \frac{1}{4}$).  To avoid repetition with the noiseless case, we omit a formal statement and derivation of this fact.  Finally, we briefly mention that BMC only has limited robustness to {\em adversarial} bit flips, since $O(w) = O(\log n)$ worst-case flips suffice to cause incorrect decisions from either the first or second batch of tests.

\section{Conclusion} \label{sec:conclusion}

We have introduced a novel scheme for sublinear-time non-adaptive group testing, and established that it attains asymptotically vanishing error probability with $t = O(k \log n)$ tests and $O(k^2 \cdot \log k \cdot \log n)$ runtime.
Our algorithm and analysis use coding-based subroutines that permit straightforward extensions to the noisy setting.

An important remaining open problem is whether the runtime can further be reduced to $k \cdot \poly(\log n)$, or better yet, to $O(k \log n)$, while still attaining $t = O(k \log n)$.  In addition, since we did not attempt to optimize constant factors, it is also of interest to sharpen the analysis (and/or modify the algorithm itself) to attain constant factors competitive with those of slower decoding techniques \cite[Ch.~2]{Ald19}.  

\appendices

\section{Concentration Inequalities} \label{sec:conc}

Throughout the paper, we make use of several standard concentration bounds for sums of independent random variables, e.g., see \cite[Sec.~4.1]{Mot10} and \cite[Ch.~2]{Bou13}.  For clarity, in this section we summarize the specific bounds used.  Letting $Z_1,\dotsc,Z_n$ be a sequence of independent and identically distributed random variables, we have the following:
\begin{itemize}
    \item (Chernoff bound) Suppose that $Z_i \in [0,1]$ almost surely, and $\EE[Z_i] = \mu$. Then, for any $\alpha > 0$, we have
        \begin{equation}
            \PP\bigg[ \sum_{i=1}^n Z_i  \ge (1+\alpha)n \mu \bigg] \le \exp\Big( -\mu n \big( (1+\alpha)\log(1+\alpha) - \alpha \big) \Big), \label{eq:strong_chernoff_1}
        \end{equation}
        and for any $\alpha \in (0,1]$, we have
        \begin{equation}
            \PP\bigg[ \sum_{i=1}^n Z_i \le (1-\alpha)n \mu \bigg] \le \exp\Big( -\mu n \big( (1-\alpha)\log(1-\alpha) + \alpha \big) \Big).
        \end{equation}
    \item (Weakened Chernoff bound) Suppose that $Z_i \in [0,1]$ almost surely, and $\EE[Z_i] = \mu$. Then, {for any $\alpha \in (0,1]$, we have
        \begin{gather}
            \PP\bigg[ \sum_{i=1}^n Z_i  \ge (1+\alpha)n \mu \bigg] \le \exp\Big( -\frac{1}{3} \alpha^2 \mu n \Big), \label{eq:weak_chernoff_1} \\
            \PP\bigg[ \sum_{i=1}^n Z_i \le (1-\alpha)n \mu \bigg] \le \exp\Big( -\frac{1}{3} \alpha^2 \mu n \Big).
        \end{gather}}
    \item (Bernstein's inequality) Suppose that $|Z_i| \le M$ almost surely, and that $\EE[Z_i] = 0$ and $\EE[Z_i^2] \le V$.  Then, for any $\delta > 0$, we have
    \begin{equation}
        \PP\bigg[ \sum_{i=1}^n Z_i  \ge t \bigg] \le \exp\bigg( - \frac{t^2}{ 2\big( n V + \frac{1}{3} M t \big) } \bigg).
    \end{equation}
\end{itemize}

\section{The Very Sparse Regime $k = O(1)$} \label{sec:fixed_k}

In our main result (Theorem \ref{the:main}), we assumed that $k \to \infty$ as $n \to \infty$.  Here we describe how to use BMC to attain $\pe \to 0$ as $n \to \infty$ in the case that $k = O(1)$, while using $t = O(\log n)$ tests and $O((\log n)^2)$ decoding time.

We again use Definition \ref{def:masking}, letting each masking string contain $w = \log n$ segments of length $4k$ and weight one, so that the total length is $t_1 = 4 k \log n$.  Similarly to Section \ref{sec:masking}, we consider the random construction of a multi-set $\Sc$ of such masking strings, with each non-zero entry of each length-$4k$ segment being independently chosen uniformly at random.  We let the size of this multi-set be $|\Sc| = \log n$.

For two such random masking strings $\sv$ and $\sv'$, the average number of collisions (i.e., $1$'s in common) follows a binomial distribution with parameters $\log n$ and $\frac{1}{4k}$, so the mean is $\frac{\log n}{4 k}$.  Hence, by the Chernoff bound (see Appendix \ref{sec:conc}), the probability of the number of collisions exceeding $\frac{\log n}{2k}$ is $O(n^{-c})$ for some $c > 0$ (here $c$ depends on $k$, but is still $\Omega(1)$ since $k = O(1)$).  By a union bound over $O( \log^2 n )$ pairs, we deduce that the probability of {\em any} two $\sv,\sv' \in \Sc$ having more than $\frac{\log n}{2k}$ collisions tends to zero as $n \to \infty$.  We henceforth condition on the (high-probability) complement of this event.

Due to this conditioning, we find that {\em any} $\sv \in \Sc$ collides with {\em any} subset $\Sctil \subseteq \Sc \setminus \{\sv\}$ of cardinality $k$ (or less) in at most $k \times \frac{\log n}{2k} = \frac{1}{2} \log n = \frac{w}{2}$ positions.  Hence, the two conditions in Definition \ref{def:lcs} hold for {\em any} $k' \le k$ distinct strings $\svtil_1,\dotsc,\svtil_{k'}$ from $\Sc$.  As a result, when we assign strings from $\Sc$ to the $n$ items uniformly at random with replacement, the only case that causes excessive collisions is that in which two defective items are assigned the same masking string.  Since $|\Sc| = \log n$ and $k = O(1)$, this occurs with probability $O\big( \frac{1}{\log n} \big)$.

Given $\Sc$ satisfying the preceding properties, the proof of Theorem \ref{the:main} goes through essentially unchanged with $w = O(\log n)$.  The number of tests is $O(w) = O(\log n)$, and the decoding time is dominated by the $O(|\Sc| \cdot \log n) = O((\log n)^2)$ term in the first batch.

\section*{Acknowledgment}

We thank Rui Zhang for helpful discussions, Sidharth Jaggi for helpful comments regarding the sublinear-time group testing literature, and Mahdi Cheragchi for helpful suggestions regarding efficient erasure-correcting codes.
This work is partly supported by the research grant MOE2017-T2-2-031 from Singapore Ministry of Education Academic Research Fund Tier-2. Binbin Chen is supported by the National Research Foundation, Prime Minister's Office, Singapore, partly under the Energy Programme administrated by the Energy Market Authority (EP Award No. NRF2017EWT-EP003-047) and partly under the Campus for Research Excellence
and Technological Enterprise (CREATE) programme. Jonathan Scarlett is supported by an NUS Early Career Research Award.

\bibliographystyle{IEEEtran}
\bibliography{JS_References}
 
\end{document}